\documentclass[a4paper,12pt,leqno]{article}

\usepackage{amssymb,amsmath}
\usepackage{amsthm}
\usepackage[abbrev]{amsrefs}
\usepackage{xy}
\xyoption{all}
\usepackage{graphicx}
\usepackage{hyperref}
\usepackage[text={14.5cm,24.5cm}]{geometry}
\usepackage{enumitem}

\numberwithin{equation}{section}

\newtheorem{defn}{Definition}[section]
\newtheorem{exm}[defn]{Example}
\newtheorem{theorem}[defn]{Theorem}
\newtheorem*{theorem*}{Theorem}

\newtheorem{lemma}[defn]{Lemma}
\newtheorem{proposition}[defn]{Proposition}
\newtheorem{corollary}[defn]{Corollary}
\newtheorem{rem}[defn]{Remark}

\newenvironment{definition}{\begin{defn}\em}{\end{defn}}
\newenvironment{example}{\begin{exm}\em}{\end{exm}}
\newenvironment{remark}{\begin{rem}\em}{\end{rem}}

\newcommand\V{\bigvee}
\newcommand\CC{\mathbb{C}}
\newcommand\RR{\mathbb{R}}
\newcommand\NN{\mathbb{N}}
\newcommand\Max{\operatorname{Max}}
\newcommand\ie{i.e.}
\newcommand\eg{e.g.}
\newcommand\st{\mid}
\newcommand\cf{\textrm{cf.}}
\newcommand\opens{\operatorname{\mathcal{O}}}
\newcommand\topology{\operatorname{\Omega}}
\newcommand\downsegment{\operatorname\downarrow}

\newcommand\SobL{\mathsf{SobL}}

\newcommand\Cstarcat{\mathsf{C^*}\textrm{-}\mathsf{Alg}}
\newcommand\MS{\mathsf{MSp}}
\newcommand\LCS{\mathsf{LocConvSp}}
\newcommand\ident{\mathrm{id}}
\newcommand\ipi{\mathcal I}

\newcommand\supp{\operatorname{supp}}
\newcommand\osupp{\operatorname{supp}^\circ}
\newcommand\Sub{\operatorname{Sub}}
\newcommand\norm[1]{\| #1\|}
\newcommand\pair{\operatorname{Pair}}
\newcommand\ps[1]{\widetilde{#1}}

\newcommand\spanmap{\operatorname{span}}

\newcommand\closedsets{\operatorname{\mathsf{C}}}
\newcommand\spinup{{\uparrow}}
\newcommand\spindown{{\downarrow}}
\newcommand\inv{*}
\newcommand\Aut{\operatorname{Aut}}
\newcommand\ms{\opens}
\newcommand\groupoid{\operatorname{\mathcal G}}
\newcommand\conv{*}

\newcommand\basechg{\boldsymbol c}

\newcommand\rs{\operatorname{R}}
\newcommand\ls{\operatorname{L}}
\newcommand\ts{\operatorname{T}}

\newcommand\slit{\boldsymbol s}
\newcommand\retraction{\sigma}

\begin{document}

\title{On the geometry of physical measurements: topological and algebraic aspects}
\author{Pedro Resende\thanks{Work funded by FCT/Portugal through project UIDB/04459/2020.}}

\date{~}

\maketitle

\begin{abstract}
We study the mathematical structure of the notion of measurement space, which extends aspects of noncommutative topology that are based on quantale theory. This yields a geometric model of physical measurements
that provides a realist picture, yet also operational, such that measurements and classical information arise interdependently as primitive concepts. A derived notion of classical observer caters for a mathematical formulation of Bohr's classical/quantum divide. Two important classes of measurement spaces are obtained, respectively from C*-algebras and from second-countable locally compact open sober topological groupoids. The latter yield measurements of classical type and relate to Schwinger's notion of selective measurement. We show that the measurement space associated to the reduced C*-algebra of any second-countable locally compact Hausdorff \'etale groupoid is canonically equipped with a  classical observer, and we establish a correspondence between properties of the observer and properties of the groupoid.\\
\vspace*{-2mm}~
\vspace*{-2mm}~
\end{abstract}

{\small \tableofcontents}

\section{Introduction}

The noncommutative topology of Giles and Kummer~\cite{GK} (see also~\cite{Giles,Akem70}) has led Mulvey~\cites{M86} to conjecture that suitable lattices of closed linear subspaces of a C*-algebra $A$ equipped with the naturally induced multiplication may play the role of spectrum of $A$. Such a notion of spectrum is in the spirit of pointfree topology~\cites{stonespaces,pointless}, according to which the main object of study is the locale of open sets of a space rather than the space of points itself. An instance of this idea is the fact that for any C*-algebra $A$ the lattice of closed ideals $I(A)$ is a locale, the induced multiplication on closed ideals coincides with intersection~\cite{MP1}, and $I(A)$ is isomorphic to the topology of the primitive spectrum of $A$. Therefore we can simply say that the locale $I(A)$ \emph{is} the spectrum of $A$. In particular, the Gelfand--Naimark representation theorem ensures that if $A$ is commutative the locale $I(A)$ contains all the information about $A$.

But there is no reason why the multiplication of more general subspaces should coincide with intersection, and Mulvey's idea was more ambitious because it aimed at finding a general definition of spectrum such that any C*-algebra --- not just commutative ones --- can be fully reconstructed from its spectrum. Such lattices equipped with multiplication generalize locales and were termed \emph{quantales} --- see~\cite{Rosenthal1} for a general introduction. This terminology was inspired by quantum physics, in particular by the idea that the multiplication can be interpreted as a temporal structure: read the product $ab$ as ``$b$ and then $a$,'' as when performing two sequential quantum measurements.

While a related operational interpretation of quantales was picked up in computer science~\cites{AV93,Re01,RV} and then in a quantale-based construction~\cite{MR} of Connes'~\cite{Connes} noncommutative space of Penrose tilings, it remained mostly dormant in the early years of quantale theory, during which special attention was given to right-sided quantales, such as those of closed right ideals of C*-algebras~\cites{BB86,BRB,Rosicky}. But for general C*-algebras these do not contain enough information and, partly inspired by Rosicky's~\cite{Rosicky} quantum frames, another attempt came again from Mulvey~\cite{Curacao}, who proposed addressing the much larger quantale $\Max A$ of all the closed linear subspaces of a C*-algebra $A$, equipped with the naturally induced involution. Although in the commutative setting this does not coincide with the localic spectrum $I(A)$ of commutative C*-algebras, and indeed this mismatch is essentially inevitable~\cite{KPRR}, the interplay between the whole quantale and its right-sided part plays an important role because the latter carries the noncommutative topology of Giles and Kummer, and from this fact and from the classification of irreducible representations of Mulvey and Pelletier~\cite{MP1} (see also~\cite{Re2forms}) follows the surprising fact that for unital C*-algebras the quantale $\Max A$ is a complete invariant~\cite{KR}.

While C*-algebras are the basic objects of noncommutative geometry in the sense of Connes, other structures play an important role, notably locally compact groupoids, in particular Lie groupoids --- for instance in modeling quotients of spaces, foliated manifolds, or spaces equipped with local symmetries ---, from which C*-algebras are then obtained as completions of convolution algebras. See~\cite{Paterson}. More generally, C*-algebras can be obtained from convolution algebras of sections of Fell bundles on groupoids~\cite{Kumjian98}, and there are important characterizations of the Fell line bundles that arise from central circle extensions of \'etale groupoids, in terms of ``diagonals'' of their C*-algebras~\cites{Kumjian,RenaultLNMath,Renault,ExePit19}.
Quantales have a close relation to groupoids, too, in fact a much better understood one than to C*-algebras because, for instance, a precise algebraic characterization exists of the quantale $\opens(G)$ associated to an \'etale groupoid $G$~\cite{Re07}, which moreover can be framed in terms of a (bi-)categorical equivalence~\cites{Re15,Funct2}. In addition, Fell bundles on $G$ are closely related to maps of quantales $p:\Max A\to\opens(G)$~\cite{QFB}.

One of the main aims of this paper is to link the above developments to the early operational interpretation of quantales in the context of quantum measurements.
More than just to develop a mathematical theory for its own sake, the main motivation is to provide an account of measurements that is also
free of the conceptual paradoxes that have plagued quantum mechanics since its beginning. This requires a conceptual shift: while states of systems are usually regarded, even if implicitly, as being the fundamental ``elements of reality,'' in this paper the elements of reality are taken to be the measurements themselves (for more of this rationale see section~\ref{sec:prelim}). In order to achieve this, the mathematical structure of spaces of measurements needs to be understood, along with its relation to state spaces. For instance, in algebraic quantum mechanics and algebraic quantum field theory systems are described by C*-algebras that contain their observables~\cite{Araki93}, and if a system is described by a C*-algebra $A$, the quantale $\Max A$, itself equipped with a suitable topology, will be regarded in this paper as being the space of measurements associated to the system described by $A$. If in turn $A$ is obtained from a Fell bundle on a groupoid $G$, the corresponding map $p:\Max A\to\opens(G)$ contains a mathematical description of how the space of measurements $\Max A$ relates to a classical state space, where the latter is the space of groupoid units $G_0$ equipped with symmetries carried by the groupoid structure.

The topology of a space of measurements $\Max A$ is not just abstract mathematical structure but rather it has its own conceptual significance: the open sets are regarded as physical properties which can be measured using only finitely many resources and carry a finite amount of classical information about the system being observed. This idea, along with an abstract definition of measurement space, was introduced in~\cite{fop}, where the intention was to exhibit a theory that takes measurements and classical information to be primitive concepts, in some sense yielding a model of Wheeler's \emph{it from bit}~\cite{itfrombit}.
The aim of the present paper is to complement this by providing both a solid introduction to the mathematics of measurement spaces and a thorough account of the main examples, occasionally delving again into the rationale behind the chosen mathematical structures, for instance showing that some axioms that are imposed (\eg, the continuity of disjunctions) are in fact not arbitrary. Alongside this, it will be seen that the topology on the quantales $\Max A$ improves the properties of the functor $\Max$ as a complete invariant of unital C*-algebras. Also, the maps of quantales $p:\Max A\to\opens(G)$ that were studied in~\cite{QFB} and are associated to Fell bundles on the groupoid $G$ will be revisited, leading to a generalization of some results of~\cite{QFB}.

In order to illustrate physically the various bits of mathematical structure that will be introduced, a running example will be that of spin measurements performed with a Stern--Gerlach apparatus, which provides a simple finite dimensional instantiation of the general constructions. The two-slit experiment is briefly addressed at the end of section~\ref{sec:classical}.

Throughout the paper I will use a few notions that are not staples of mathematical physics. For general information I direct the reader to \cite{bigdomainbook,stonespaces,picadopultr} for continuous lattices and the Scott topology, \cite{elephant}*{\S C1.2} for sober spaces, \cites{stonespaces,picadopultr} and ~\cite{elephant}*{\S C1.1--C1.2} for locales, \cite{Rosenthal1} for quantales, and~\cites{Re07,QFB,SGQ} for the general interplay between \'etale groupoids, quantales, and C*-algebras.

\section{Measurements and observable properties}\label{sec:prelim}

This non-technical section can be safely skipped by those who are mostly interested in the mathematical development. In it I recall (and sometimes improve) the rationale behind spaces of measurements, leading to the claim that (abstract) measurements are the points of a sober topological space whose open sets are the observable properties. Each observable property can be thought of as corresponding to a finite chunk of classical information, albeit of a ``semi-decidable'' type: not obtaining such information does not mean that there is a complementary observable property. This interpretation of open sets originated in computer science \cites{topologyvialogic,Smyth83,Ab91}, with the difference that spaces of measurements are not spaces of states, and the interpretation of $m\in U$ for a measurement $m$ and an open set $U$ is not the same as the interpretation of $x\in U$ for a state $x$.

\subsection{Rationale}

States of systems in classical physics are abstractly identified with sets of properties, such as position and momentum, and physical laws are expressed in terms of algebraic and geometric regularities to be found in spaces of such abstract states. It is implicit in this idea that physical systems exist ``out there'' in their own right and do not depend on how we observe or describe them. Quantum physics has challenged this view and led to a large variety of interpretations, ranging from those that regard a wave function as being ``real'' (a physical state of a system) to the opposite extreme where it is just viewed as a mathematical bookkeeping tool for predicting outcomes of observations, or even just a state of knowledge of an observer (see~\cite{PBRtheorem} and references therein). 

Nevertheless, the ontic view of the wave function, along with its ``collapse,'' has become pragmatic common ground for many physicists, not least due to the contribution of von Neumann, whose quantum mechanical description of systems entangled with experimental devices and observers led him to the projection postulate and to discussing conditions under which collapse occurs, to the extent of invoking the subjective perception of observers~\cite{vonNeumann1955}*{Ch.\ VI}. This has created difficulties, not only due to the conspicuous absence of a scientific understanding of subjective perception in a fundamental sense, but also because the definition of the boundary between the observed system and the observer is elusive. Bell called this boundary the ``shifty split'' and claimed, somewhat provocatively, that ``measurement'' ought to be banned from the language of quantum mechanics~\cite{Bell90}. While it is true that decoherence
has produced a more perspicuous understanding of how phenomena appear to become classical, and even of why certain basis states are selected, the basic problem persists: the experimental observation of a single outcome in the lab is not explained by Schr\"odinger's equation~\cite{adler}, so the need remains for interpretations such as Everett's many-worlds,
or, alternatively, modifications like stochastic
or gravity-induced
collapse.

Part of the problem is likely one of language, for it is well known that the language which we use to describe concepts helps shape those same concepts. This is true in everyday language, and when it comes to physics it becomes a matter of mathematical language: undoubtedly the existence of a sophisticated geometry with which to model space and time has helped carve our understanding of classical mechanics along with the feeling that the concepts associated with it are natural and ``real.'' However, the usual abstract notions of state are derived quite indirectly from our experience of physical phenomena, and it can be argued that measurements themselves are closer to direct experience. 
It is therefore justified that attention should be given to geometric properties of measurements.
For instance, borrowing a nice illustration from~\cite{PhillipBall}*{p.~88} (\cf\ \cite{Omnes1994}*{Ch.\ 5}), the idea that a photon seems to travel along a straight path in a lab follows not from any single measurement but from the fact that if we try to detect a photon in the space between a fixed source and target we can only find it along the straight path between the two and not anywhere else. But each such detection of the photon, \ie, at each of the possible positions along that straight path, is in fact a different measurement, made using a slightly different apparatus, so the appearance of a straight path \emph{hinges on geometric structure which is found in the collection of photon detecting measurements, rather than on a property of an independent entity called a photon} --- see Fig.\ \ref{fig:path}.
\begin{figure}[h!]
\begin{center}\includegraphics[width=\textwidth]{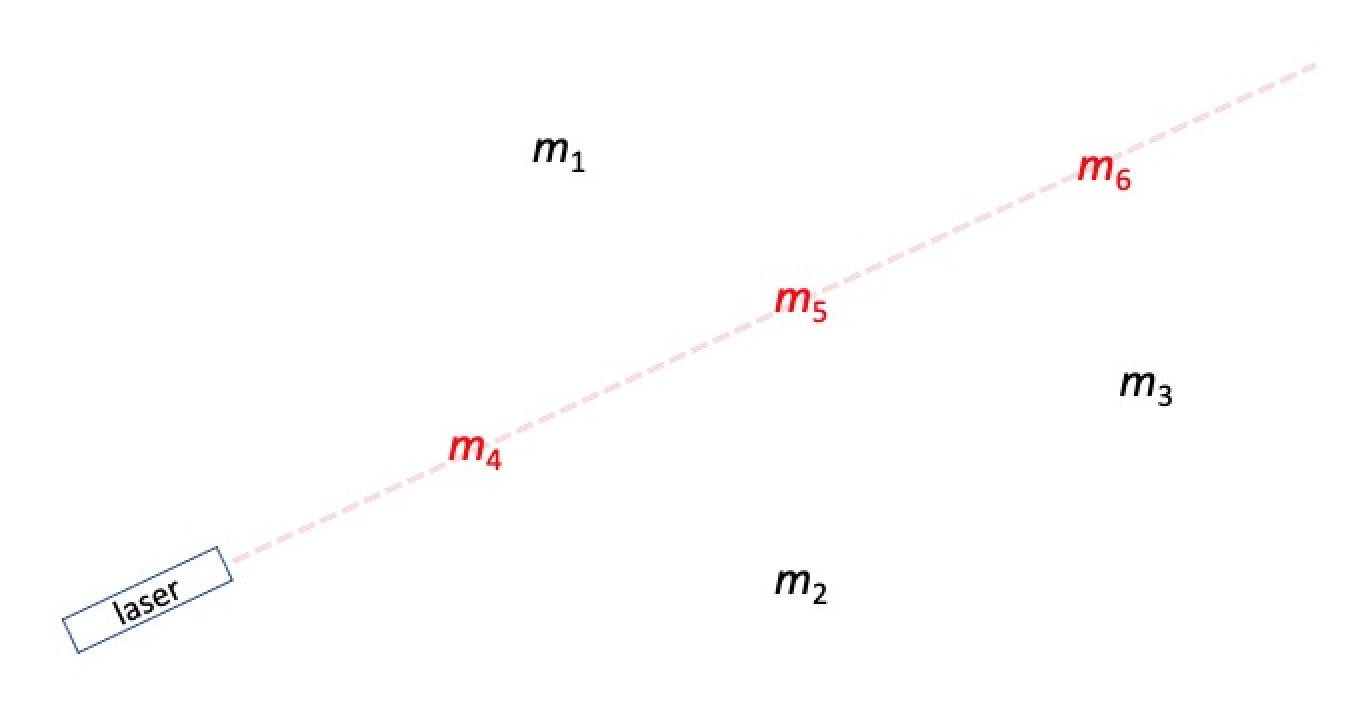}\caption{The ``straight path of a photon'' is a path in a space of measurements.}\label{fig:path}\end{center}
\end{figure}

However, reliance on concrete definitions of experimental procedure and apparatus is not a mathematical approach of the kind that would enable us to describe laws of physics in terms of mathematical regularities in a putative space of measurements, so an abstract definition is needed according to which measurements are as abstract as states are in classical physics. The analog of the idea that states are identified with their sets of properties can be that measurements correspond to sets of \emph{observable properties}, namely those of the kind that can be elicited in the course of a measuring process which, intuitively speaking, consumes only finite time and resources and furthermore meets Bohr's demands that the results of measurements ought to be communicable in the language of classical physics.
More precisely, each observable property can be regarded as consisting of a finite amount of \emph{classical information}, informally understood to be that which is recorded by writing a finite list of 0s and 1s on a sheet of paper. For instance, a single run of a typical Stern--Gerlach experiment for a spin 1/2 particle yields a single bit of classical information, corresponding to the two measurable deflections.
Physical measurements which are equivalent in terms of the information they produce are representations of the same \emph{abstract measurement}.

\subsection{Topology}

Suppose $E$ is the collection of all the physical experiments which can be performed on a given physical system. Such a collection is not abstract at all, since a typical member of $E$ will involve an apparatus, an experimental procedure, a list of possible outcomes, etc. However, as suggested above, the ``meaning'' of each experiment boils down to the collection of finite chunks of classical information which can be extracted from it, each of which is regarded as an \emph{observable property} of the system.

Each observable property can be identified with the set $U\subset E$ of the experiments via which that property can be asserted. More precisely, the meaning of $e\in U$ is that \emph{the experiment $e$ is compatible with $U$} in the sense that performing $e$ may give us the information which is represented by $U$.
For instance, let $\boldsymbol z$ be a Stern--Gerlach experiment that measures the spin of an electron along the $z$ axis, and let $U^\uparrow$ represent the information that there was an upwards deflection. So $U^\uparrow$ is compatible with $\boldsymbol z$ because it represents information that \emph{can be} obtained (but is not necessarily obtained) by performing $\boldsymbol z$. Equally, $\boldsymbol z$ is compatible with the information $U^\downarrow$ that there was a downwards deflection, so we have $\boldsymbol z\in U^\uparrow\cap U^\downarrow$ (Fig.\ \ref{fig:SG}).
\begin{figure}[h!]
\begin{center}\includegraphics[width=0.5\textwidth]{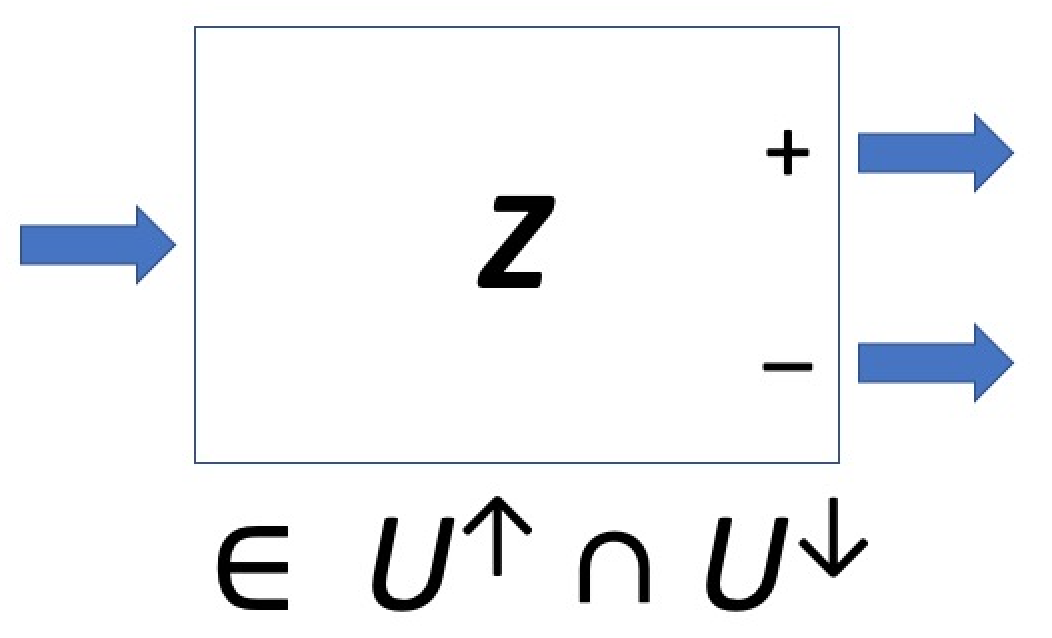}\caption{A schematic Stern--Gerlach apparatus. Upwards deflection corresponds to recording the observable property $U^\spinup$, and downwards deflection corresponds to recording the observable property $U^\spindown$.}\label{fig:SG}\end{center}
\end{figure}
Note that an inclusion $V\subset U$ of observable properties means that $V$ represents \emph{more} information than $U$. This can be read as the logical proposition ``$V$ implies $U$.'' For instance, $U^\uparrow\cap U^\downarrow$ represents more information than $U^\downarrow$ because the former corresponds to a situation where \emph{more} possible outcomes of a spin measurement (both up and down) have been shown to be possible (for instance by running the same experiment twice), whereas $U^\downarrow$ is only the information that a downwards deflection is possible. Accordingly, we may interpret intersections and unions as logical operations:
\begin{itemize}
\item $U\cap V$ is the \emph{conjunction} of $U$ and $V$, which can be verified by conducting the same experiment twice, each time recording one of the properties $U$ and $V$ --- for instance, performing the Stern--Gerlach experiment $\boldsymbol z$ at least twice we can verify that $\boldsymbol z\in U^\uparrow\cap U^\downarrow$.
\item $U\cup V$ is the \emph{disjunction} of $U$ and $V$, which can be verified by performing an experiment once and recording one of the properties $U$ or $V$.
\end{itemize}
But there is an important asymmetry. A disjunction $\bigcup_i U_i$ of infinitely many properties can equally well be verified by performing an experiment just once, whereas a conjunction of infinitely many properties $U_i$ may require us to perform infinitely many experiments, which can take infinitely much time or use infinitely many resources. So, in general, $\bigcap_i U_i$ might not be an observable property. This leads to the idea that the observable properties should be the open sets of a topology on $E$, with $\emptyset$ and $E$ being the \emph{impossible property} and the \emph{trivial property}, respectively.
Then any two experiments which belong to the same open sets yield exactly the same classical information, so they should be considered equivalent for all practical purposes. Each equivalence class is therefore an \emph{abstract experiment}, which is what in this paper will be meant by a \emph{measurement}.

Accordingly, from here on any topological space of measurements $M$ will be required to be $T_0$; that is, the specialization order $\le$ of $M$, which is defined by
\[
m\le n\iff m\in\overline{\{n\}},
\]
is a partial order. Moreover, following~\cite{fop}, a criterium of logical completeness implies that $M$ should be a sober space (any collection of open sets that ``looks'' like the set of open neighborhoods of a point must indeed be the set of open neighborhoods of a point), and thus (i)~the specialization order is directed complete (a dcpo); and (ii)~the topology is contained in the Scott topology, whose open sets $U$ are the upwards closed sets such that, for all directed sets $D\subset M$, if $\V D\in U$ then $d\in U$ for some $d\in D$. See~\cite{fop} for a more in-depth discussion.

An additional property which will often hold is separability. This makes sense if one thinks of the physical experiments that can be given a formal description by, say, writing down their specifications on paper using any finite alphabet. The collection of all such specifications is denumerable, and thus the collection of all the conceivable human-made experiments is at most denumerable. Hence, a non-separable space of measurements would be one for which there exists an open set (an observable property) that does not contain any human-made measurement. Such an observable property could only be asserted on the basis of some phenomenon which is beyond the control of any experimenter. Therefore, imposing separability is a way of formalizing the assumption that any observable property whatsoever is ultimately detectable by means of a human-made experiment.

\subsection{Continuity, contexts and observers}\label{sec:contobs}

Let $M$ and $N$ be sober spaces of measurements. A continuous map
\[
f:M\to N
\]
can be regarded as a \emph{representation} of $M$ in terms of the measurements in $N$. Since the open sets of $M$ and $N$ are understood to be \emph{finitely} observable properties, the natural way to interpret the continuity of $f$ is that a \emph{finite process} exists for transforming $m$ into $f(m)$, and therefore that observing a property of $f(m)$ by finite means is a particular way of obtaining a property of $m$ by finite means: the property $f^{-1}(U)$ of $m$ is \emph{the property that $U$ can be observed after we apply the transformation $f$ to $m$}.

In particular, if $\iota:S\to M$ is a topological embedding, we can identify $S$ with the subspace $\iota(S)\subset M$ and $\iota$ with the inclusion map, and regard $S$ as a selected collection of measurements, such as when restricting to spin measurements along $z$, or paying attention to Alice's measurements while ignoring Bob's measurements, etc. This provides grounds for addressing \emph{contextuality}, with $S$ being regarded as a \emph{context} within $M$.

Now let us consider topological retracts. Let $\opens\subset M$ be a subspace together with a continuous map $r:M\to \opens$ such that $r(\omega)=\omega$ for all $\omega\in \opens$. Here $\opens$ is a context but, in addition, every measurement $m\in M$ can be represented continuously in $\opens$.
This provides the grounds for defining a notion of \emph{observer context} of $M$: a context $\opens$ equipped with a finite \emph{transformation recipe} for describing, in the language afforded by $\opens$, any \emph{event} that occurs in the system being measured, furthermore such that if $\omega$ is already in $\opens$ then the transformation is trivial because $r(\omega)=\omega$.

\section{Disjunctions of measurements}\label{sec:disjunctions}

In this section I will further address the topology of spaces of measurements by requiring, in addition to sobriety, that an operation of logical disjunction of measurements should exist. This has been done preliminarily in~\cite{fop}, but most mathematical details were glossed over. The disjunction operation has a clear interpretation in terms of the classical information associated to measurements, and it will be deduced that it should coincide with the operation of binary join in the specialization order of the topology, and moreover that this operation should be continuous, thus leading to the definition of \emph{sober lattice}.

In many situations the physical realization of a disjunction is easy to describe, such as when regarding a spin measurement to be a disjunction of the various possible outcomes.
Regardless of whether or not we know how to physically realize disjunctions in any given practical situation, their existence in spaces of measurements is, from the mathematical point of view, a useful completeness requirement which has the consequence that spaces of measurements are complete lattices and that the join operation in any arity is continuous. As we shall see, this turns spaces of measurements into topological sup-lattices, whose localic analogues have been studied in~\cite{RV}. 

\subsection{Disjunctions = continuous joins}\label{sec:continuousjoins}

Consider the two variants of the measurement of a spin-1/2 particle along $z$, as depicted in Fig.\ \ref{fig:zupzdown}.
\begin{figure}[h!]
\begin{center}\includegraphics[width=0.4\textwidth]{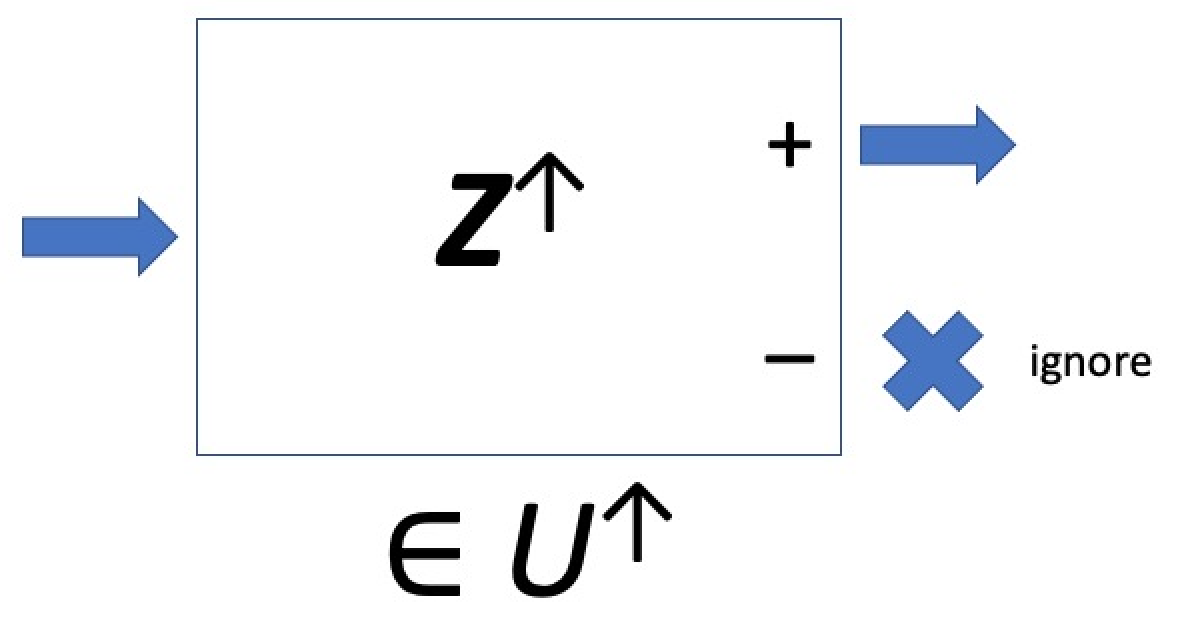}\quad\quad\quad\includegraphics[width=0.4\textwidth]{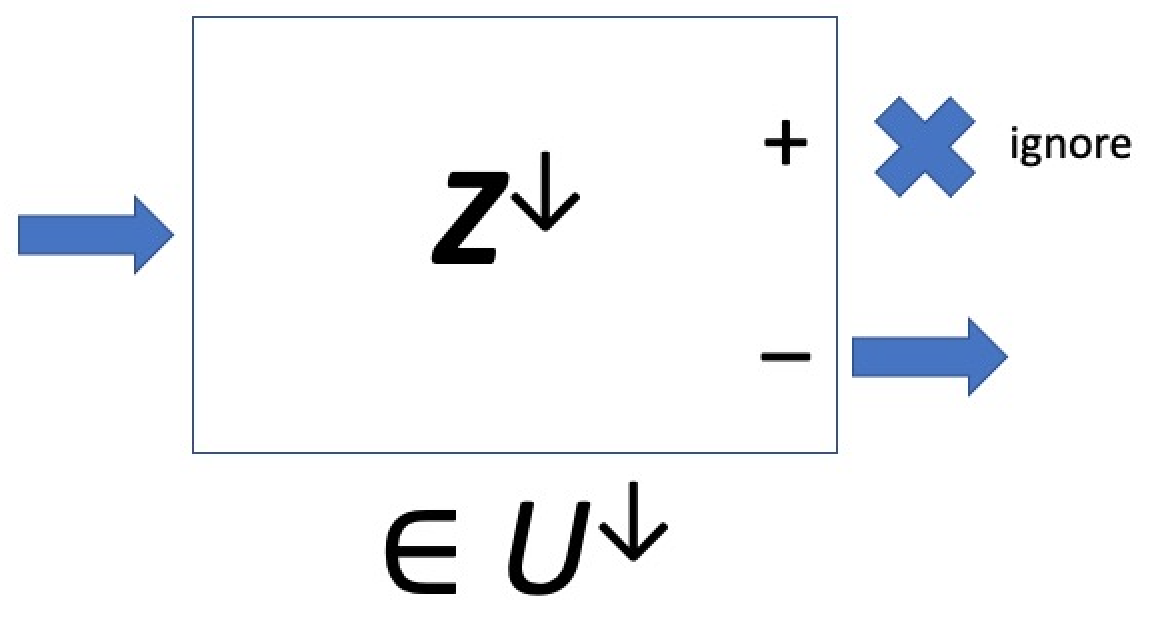}\caption{
The measurement $\boldsymbol z^\spinup$ is compatible with the observable property $U^\spinup$ but not with $U^\spindown$, and $\boldsymbol z^\spindown$ is only compatible with $U^\spindown$.
}\label{fig:zupzdown}\end{center}
\end{figure}
The measurement $\boldsymbol z$ of Fig.\ \ref{fig:SG} is ``richer'' in the sense that both $U^\spinup$ and $U^\spindown$ can be read from it, whereas only $U^\spinup$ can be read from $\boldsymbol z^\spinup$ and only $U^\spindown$ can be read from $\boldsymbol z^\spindown$. In other words, $\boldsymbol z$ is compatible with \emph{more} observable properties than either $\boldsymbol z^\spinup$ or $\boldsymbol z^\spindown$. This makes it a \emph{less determined} measurement than either of the other two. In fact a run of $\boldsymbol z$ can be said to be either a run of $\boldsymbol z^\spinup$ or a run of $\boldsymbol z^\spindown$, and thus any property that can be observed in the course of a single run of $\boldsymbol z$ should be observable by running $\boldsymbol z^\spinup$ once or by running $\boldsymbol z^\spindown$ once. So we can regard $\boldsymbol z$ as a logical \emph{disjunction} of $\boldsymbol z^\spinup$ and $\boldsymbol z^\spindown$, which we can write as follows:
\[
\boldsymbol z=\boldsymbol z^\spinup\vee \boldsymbol z^\spindown.
\]
Let us examine this idea a bit further. If $V$ is compatible with $\boldsymbol z^\spinup$ then it must also be compatible with $\boldsymbol z$, since the former is a particular way of doing the latter. If, in addition, $W$ is compatible with $\boldsymbol z^\spindown$, then $V\cap W$ is compatible with $\boldsymbol z$ --- this can be verified by running $\boldsymbol z$ at least twice. But not all the observable properties of $\boldsymbol z$ should be intersections of this form. Rather, whatever the observable properties of $\boldsymbol z$ are, the finite information they represent should be contained in information that is obtained from multiple runs of $\boldsymbol z^\spinup$ or $\boldsymbol z^\spindown$. In other words, any property $U$ which is compatible with $\boldsymbol z$ must be implied by an intersection $V\cap W$ of properties such that $V$ is compatible with $\boldsymbol z^\spinup$ and $W$ is compatible with $\boldsymbol z^\spindown$.

Based on these ideas, let us give a precise definition of disjunction:
\begin{definition}
Let $m$ and $n$ be measurements. A measurement $m\vee n$ is a \emph{disjunction} of $m$ and $n$ if the following conditions hold:
\begin{enumerate}
\item\label{dis1} Any observable property which is compatible with $m$ is also compatible with $m\vee n$.
\item\label{dis2} Any observable property which is compatible with $n$ is also compatible with $m\vee n$.
\item\label{dis3} If an observable property $U$ is compatible with $m\vee n$ then there are observable properties $V$ and $W$, respectively compatible with $m$ and $n$, such that $V\cap W\subset U$.
\end{enumerate}
\end{definition}

As stated above, it will be useful to assume that every pair of measurements has a disjunction. The consequences of this requirement are summarized as follows:

\begin{theorem}
Let $M$ be a sober space of measurements. The following conditions are equivalent:
\begin{enumerate}
\item\label{thm:dis1} Every pair $m,n\in M$ has a disjunction $m\vee n$.
\item\label{thm:dis2} Every pair $m,n\in M$ has a join $m\vee n$ in the specialization order of $M$, and the operation ${\vee}:M\times M\to M$ is continuous.
\end{enumerate}
Moreover, when the above equivalent conditions hold the disjunction of $m$ and $n$ is unique and it coincides with the join $m\vee n$.
\end{theorem}

\begin{proof}
Let us prove that \eqref{thm:dis1} implies \eqref{thm:dis2}. Assume \eqref{thm:dis1}. The first two conditions in the definition of disjunction are immediately equivalent to the statements that $m\le m\vee n$ and $n\le m\vee n$ in the specialization order. Let $u$ be another upper bound of $m$ and $n$, and let $U$ be an open set containing $m\vee n$. The third condition in the definition of disjunction implies that there are open sets $V$ and $W$ such that $m\in V$, $n\in W$, and $V\cap W\subset U$, from which it follows that $u\in U$. Hence, we have $m\vee n\le u$, and thus $m\vee n$ is the join of $m$ and $n$ (unique because $M$ is a $T_0$ space). So we have concluded from \eqref{thm:dis1} that $M$ is a join semilattice.

It is easy to see that the intersection of each pair of open sets $V$ and $W$ coincides with their pointwise join:
\[
V\cap W = V\vee W := \{m\vee n\st m\in V\text{ and }n\in W\}.
\]
Therefore the third condition in the definition of disjunction translates precisely to the statement that $\vee:M\times M\to M$ is a continuous operation at the point $(m,n)$, so \eqref{thm:dis1} implies that $\vee$ is everywhere continuous.

Now let us prove that \eqref{thm:dis2} implies \eqref{thm:dis1}. Assume \eqref{thm:dis2}, and let $m,n\in M$. We prove that $m$ and $n$ have a disjunction by proving that the join $m\vee n$ is a disjunction. The first two conditions of the definition of disjunction are immediate because they are equivalent to the statement that $m\vee n$ is an upper bound of $\{m,n\}$. For the third condition, let $U$ be an open set and $m\vee n\in U$.
The continuity of $\vee$ ensures that there is a basic open set $V\times W$ of $M\times M$ containing $(m,n)$ such that the pointwise join $V\vee W$ is contained in $U$, which, due to the equality $V\cap W=V\vee W$, is precisely the third condition in the definition of disjunction. So the join $m\vee n$ is the (necessarily unique) disjunction of $m$ and $n$.
\end{proof}

In~\cite{fop} it is further required that the specialization order of any sober space of measurements should have a least element (the \emph{impossible measurement}), hence making it a join semi-lattice with zero (in fact a complete lattice, as we shall see below).

\subsection{Sober lattices}

Let us study the mathematical properties of the spaces of measurements that we have been discussing, in particular requiring every pair of measurements to have a disjunction:

\begin{definition} \cite{fop}
By a \emph{sober lattice} is meant a sober topological space $L$ whose specialization order has a least element $0$ and a join $x\vee y$ of each pair of elements $x,y\in L$, such that the binary join operation
\[
\vee:L\times L\to L
\]
is continuous.
A \emph{homomorphism} of sober lattices $f:L\to M$ is a continuous map that preserves binary joins and $0$:
\begin{itemize}
\item $f(x\vee y)=f(x)\vee f(y)$ for all $x,y\in L$;
\item $f(0)=0$.
\end{itemize}
The resulting \emph{category of sober lattices} is denoted by $\SobL$. An \emph{embedding} of sober lattices is a homomorphism which is also a topological embedding.
\end{definition}

\begin{example}\label{exm:2}
The two element chain $\boldsymbol 2:=\{0,1\}$ ($0<1$) is a sober lattice with the Alexandrov topology $\topology(\boldsymbol 2)=\bigl\{\emptyset,\{1\},\{0,1\}\bigr\}$. Regarded as a space of measurements, it contains a unique non-impossible measurement, $1$, which can be thought of as the mere detection of a system without asserting any additional properties.
\end{example}

\begin{theorem}
Let $L$ and $M$ be sober lattices. The cartesian product $L\times M$, equipped with the product topology, is a sober lattice. Moreover, it is the biproduct of $L$ and $M$ in $\SobL$ (so we also refer to $L\times M$ as the \emph{direct sum} of $L$ and $M$, and denote it by $L\oplus M$).
\end{theorem}

\begin{proof}
The binary joins in the pairwise order of $L\times M$ are computed as
\[
(x,y)\vee(z,w) = (x\vee z,y\vee w),
\]
they are continuous in the product topology, there is a least element $(0,0)$, and the product of sober spaces is sober (the latter follows from the fact that the spectrum functor from locales to sober spaces is a right adjoint and therefore preserves products, \cf\ \cite{stonespaces}). Moreover, it is well known~\cite{JT} that $L\times M$ is a biproduct in the category of sup-lattices, whose projections and coprojections
\[
\xymatrix{
L\ar@<1ex>[rr]^{i_1}&&L\times M\ar@<1ex>[ll]^{\pi_1}\ar@<1ex>[rr]^{\pi_2}&&M\ar@<1ex>[ll]^{i_2}
}
\]
are defined by $\pi_1(x,y)=x$, $\pi_2(x,y)=y$, $i_1(x)=(x,0)$, and $i_2(y)=(0,y)$. Then, since $L\times M$ is a categorical product of topological spaces, it is a categorical product of sober lattices. In order to see that it is a coproduct in $\SobL$ it suffices to observe that the coprojections are obviously continuous, and that for any pair of homomorphisms of sober lattices
\[
\xymatrix{
L\ar[rr]^f&&N&& M\ar[ll]_g
}
\]
the copairing map $h=[f,g]:L\times M\to N$, which is defined by $h(x,y)=f(x)\vee g(y)$, is continuous.
\end{proof}

\begin{example}\label{exm:22}
Consider the chain $\boldsymbol 2$ of Example~\ref{exm:2}. The product $\boldsymbol 2^2:=\boldsymbol 2\times\boldsymbol 2$ is the following four element Boolean algebra:
\[
\xymatrix{
&(1,1)\ar@{-}[dl]\ar@{-}[dr]\\
(1,0)\ar@{-}[dr]&&(0,1)\ar@{-}[dl]\\
&(0,0).
}
\]
This can be regarded as the space of measurements of a classical bit, where for instance $(1,0)$ is ``one'' and $(0,1)$ is ``zero.'' Similarly, the Boolean algebra $\boldsymbol 2^n$ is a model of the measurements on a classical $n$ state system (however not containing every possible measurement on such a system, but rather only measurements of ``static type'' --- \cf\ section~\ref{sec:groupoids}).
\end{example}

\begin{theorem}\label{arbitrarycontinuousjoins}
The specialization order of any sober lattice $L$ is a sup-lattice, and for any set $J$ the join operation $\V:L^J\to L$ is continuous with respect to the product topology of $L^J$. Moreover, homomorphisms of sober lattices preserve arbitrary joins.
\end{theorem}

\begin{proof}
Sober spaces have joins of directed sets, and thus the existence of $0$ (the $0$-ary join) and of binary joins implies that all joins exist: given a set $S\subset L$, the join $\V S$ equals $\V \check S$, where $\check S$ is the directed set
\[
\check S=\{0\}\cup S\cup\{x\vee y\st x,y\in S\}\cup\{x\vee y\vee z\st x,y,z\in S\}\cup\ldots
\]
Moreover, continuous maps of sober spaces preserve directed joins, and thus a homomorphism of sober lattices preserves all joins. It remains to verify that joins are continuous for arbitrary arities.
For $J$ of cardinality $0$ or $1$ this is trivial, and for finite cardinalities $\ge 2$ the continuity of joins is obtained by iterating binary joins. So let us consider an infinite set $J$, and write $\check J$ for the set of finite subsets of $J$. For each family $x\in L^J$ and for each $F\in\check J$ write $\check x_F$ for the finite join $\V_{j\in F} x_j$ (in particular $\check x_\emptyset=0$). This defines a family $\check x\in L^{\check J}$ such that
\[
\V_{F\in\check J} \check x_F = \V_{j\in J} x_j.
\]
In order to prove the continuity of $J$-indexed joins, let $\V_{j\in J} x_j\in U$ for an open $U\subset L$. Then $\V_{F\in\check J} \check x_F\in U$ and, since the set $\{\check x_F\st F\in\check J\}$ is directed and $U$ is Scott open (because $L$ is sober), there is $F\in\check J$ such that $\check x_F\in U$. Since the join operation is continuous for all finite arities, for each $j\in F$ there is an open $U_j\subset L$ such that $\V_{j\in F}y_j\in U$ for all $y\in\prod_{j\in F}U_j$. Now consider the basic open $V:=\prod_{j\in J} V_j$ of $L^J$ such that $V_j=U_j$ for all $j\in F$, and $V_j=L$ for $j\notin F$, and let $y\in V$. Then $\V_{j\in J} y_j\ge \V_{j\in F}y_j\in U$, and thus
$\V_{j\in J} y_j\in U$ because $U$ is upper-closed. Therefore $\V:L^J\to L$ is continuous.
\end{proof}

The following is a simple but useful property:

\begin{lemma}\label{tinyprop}
Let $L$ be a sober lattice, and $S\subset L$ any set. Then
\[
\V S = \V\overline S.
\]
\end{lemma}

\begin{proof}
For any $s\in S$ we have $s\le\V S$, so $S\subset\downsegment(\V S)$. Then, since $\downsegment(\V S)$ is closed (it is the closure of the singleton $\{\V S\}$), we also have $\overline S\subset\downsegment(\V S)$, which implies
$\V\overline S\le\V S$. The converse inequality is trivial.
\end{proof}

We conclude this introduction to sober lattices by pointing out that the requirement of continuity of binary joins is not trivial:

\begin{example}
Not all sober spaces are sober lattices, even if their specialization orders have all joins. An example of this is the space
\[
X = \{0\}\cup (\{0\}\times [0,1])\cup\{2\}\cup(\{1\}\times[0,1])\cup\{1\},
\]
whose specialization order is the least partial order that satisfies the following conditions:
\begin{itemize}
\item $0$ is the least element and $1$ is the top element;
\item $(0,x)<(1,x)$ for all $x\in[0,1]$;
\item $2<(1,x)$ for all $x\in[0,1]$.
\end{itemize}
This is an infinite sup-lattice, which can be represented as follows:
\[
\xymatrix{
&&1\ar@{-}[dll]\ar@{-}[dl]\ar@{-}[d]\\
(1,0)\ar@{-}[d]\ar@{-}[drrr]\ar@{.}[r]&(1,x)\ar@{-}[drr]\ar@{-}[d]\ar@{.}[r]&(1,1)\ar@{-}[d]\ar@{-}[dr]\\
(0,0)\ar@{-}[drr]\ar@{.}[r]&(0,x)\ar@{-}[dr]\ar@{.}[r]&(0,1)\ar@{-}[d]&2\ar@{-}[dl]\\
&&0
}
\]
For instance, we have the following joins:
\begin{itemize}
\item $(m,x)\vee(n,y)=1$ if $x\neq y$;
\item $(m,x)\vee 2=(1,x)$.
\end{itemize}
A basis for the topology consists of the following sets:
\begin{itemize}
\item $X$;
\item $\{(1,x),1\}$ for each $x\in[0,1]$;
\item $(\{0,1\}\times U)\cup\{1\}$ for each open set $U$ of the real interval $[0,1]$;
\item $\{2\}\cup(\{1\}\times[0,1])\cup\{1\}$ (the principal filter generated by $2$);
\item $\{1\}$.
\end{itemize}
(Notice that the subspace $\{0\}\times[0,1]$ is homeomorphic to the real interval $[0,1]$, whereas $\{1\}\times[0,1]$ is discrete.) It is not hard to see that the only irreducible closed sets of $X$ are the principal ideals, so $X$ is a sober space. However, it is not a sober lattice: we have $(1,1)=(0,1)\vee 2$, and $(1,1)$ is contained in the open set $U=\{(1,1),1\}$; but the only open set $V$ containing $2$ contains all the points $(1,x)$, and any open set $W$ containing $(0,1)$ must contain elements $(1,x)$ with $x<1$, so the intersection $V\cap W$ contains such elements, too, showing that $V\cap W\not\subset U$. Hence, the binary join operation is not continuous at the point $((0,1),2)$ [in fact it is not continuous at any point $((0,x),2)$].
\end{example}

\begin{remark}
Let $L$ be a sober lattice and $S\subset L$ a subset, possibly infinite. Analogously to the interpretation of $x\vee y$ as the disjunction of $x,y\in L$, one may ask whether the join $\V S$ can be regarded as the disjunction of all the measurements in $S$. The rationale given above for binary disjunctions can of course be iterated for arbitrary finitary disjunctions, namely, in the absence of more specific information, asserting an observable property $U$ of $m:=x_1\vee\ldots\vee x_k$ may need at least $k$ runs of the measurement $m$: one run for measuring a property $V_i$ containing $x_i$ for each $i\in\{1,\ldots,k\}$, such that $\bigcap_{i=1}^k V_i\subset U$. If $S$ has infinite order the join $\V S$ can still be interpreted as a disjunction, but one does not need infinitely many runs of the measurement $\V S$ (which would be inconsistent with the requirement of finite observability of $U$): writing $\check S$ for the directed set
\[
\check S := \{\V F\st F\subset S\text{ and }F\text{ is finite}\},
\]
(\cf\ Theorem~\ref{arbitrarycontinuousjoins}), we have
\[
\V S = \V \check S
\]
and thus, since the topology of any sober lattice is contained in the Scott topology, the condition $\V S\in U$ is equivalent to the existence of a finite set $F\subset S$ such that $\V F\in U$, so any property of $\V S$ can be measured by running a finite disjunction $\V F$ of elements of $S$.
\end{remark}

\subsection{Quantum superpositions}

\paragraph{Linear powerspaces.}

Let us get a little closer to the standard formalism of quantum mechanics by considering quantum superpositions defined in complex vector spaces. In what follows, $A$ will be an arbitrary locally convex space. For the purposes of this section, $A$ can be thought of as being a Hilbert space, although in section~\ref{sec:Cstar} the natural examples will be seen to be C*-algebras (\cf\ Example~\ref{exm:hilbrep} and Example~\ref{exm:Cstarrep} below). Hence, borrowing the notation of~\cites{Curacao,MP1}, let us denote by $\Max A$ the set of all the closed linear subspaces of $A$.
This will always be assumed to be a topological space whose topology is the lower Vietoris topology, which is generated by a subbasis formed by the sets of the form
\[
\ps U = \{P\in\Max A\st P\cap U\neq\emptyset\}
\]
where $U$ is an open set of $A$.

\begin{theorem}
If $A$ is a locally convex space then $\Max A$ is a sober lattice whose specialization order is the inclusion order. Moreover, if $A$ is separable so is $\Max A$.
\end{theorem}

\begin{proof}
$\Max A$ is a sober space for any locally convex space $A$, by~\cite{RS2}.
Let us see that the specialization order $\le$ is the inclusion order. Let $P,Q\in\Max A$, and assume $P\subset Q$. In order to conclude $P\le Q$ we need to show that for all open sets $\mathcal U\subset\Max A$ the condition $P\in\mathcal U$ implies $Q\in\mathcal U$. It suffices to do this for a subbasic open set $\ps U$. Then the condition $P\in\ps U$ means that $P\cap U\neq\emptyset$, which implies $Q\cap U\neq\emptyset$, so $Q\in\ps U$ and we conclude $P\le Q$.

Now let us prove the converse. Assume $P\le Q$, and let $p\in P$. For any open set $U\subset A$ containing $p$ we have $P\in\ps U$, which implies $Q\in\ps U$, and thus $U\cap Q\neq\emptyset$. Hence, $p\in \overline Q=Q$, and thus $P\subset Q$.

So the specialization order is a sup-lattice, and all that is left to prove is that $\vee$ is continuous. Let $P,Q\in\Max A$ and $U$ an open set of $A$ such that $P\vee Q\in\ps U$. We have $P\vee Q=\overline{P+Q}$,
so there are $p\in P$ and $q\in Q$ such that
$p+q\in U$. By the continuity of the sum in $A$, there are neighborhoods
$V\ni p$ and $W\ni q$ such that the pointwise sum $V+W$ is contained in $U$. Then $\ps V\times\ps W$ is an open set of $\Max A\times\Max A$ that contains $(P,Q)$, and for all $(P',Q')\in\ps V\times\ps W$ there are $p'\in P'\cap V$ and $q'\in Q'\cap W$, so $p'+q'\in U$. Hence, $P'\vee Q'\in\ps U$, and we conclude that $\vee$ is continuous.

Finally, let us verify the statement about separability. Assume that $A$ is separable, and let $S\subset A$ be a countable dense subset. Let $\mathcal U$ be an open set of $\Max A$. It suffices to take $n\in\NN$ and $\mathcal U=\ps U_1\cap\ldots\cap\ps U_n$ for some open sets $U_1,\ldots, U_n\subset A$. For each $U_i$ let $a_i\in S\cap U_i$. Then $\langle a_1,\ldots, a_n\rangle\in\mathcal U$, and thus $\Max A$ is separable because the set of spans of finite subsets of $S$ is countable.
\end{proof}

Let us fix some terminology:

\begin{definition}
Let $A$ be a locally convex space. The sober lattice $\Max A$ is referred to as the \emph{linear powerspace} of $A$.
\end{definition}

Let $A$ be a locally convex space, and $a,b\in A$. The spans $\langle a\rangle$ and $\langle b\rangle$ are elements of $\Max A$, so let us think of them as being measurements. In $\Max A$ their join is the two-dimensional space $\langle a,b\rangle$. The linear combinations $\alpha a+\beta b\in \langle a,b\rangle$ with $\alpha,\beta\in\CC$ can be regarded as (non-normalized) quantum superpositions of $a$ and $b$.

\begin{example}\label{exm:hilbrep}
Consider the measurements $\boldsymbol z^\spinup$, $\boldsymbol z^\spindown$, and $\boldsymbol z=\boldsymbol z^\spinup\vee \boldsymbol z^\spindown$ described in section~\ref{sec:continuousjoins}. These can be represented by means of the usual Hilbert space representation of a qubit, namely by taking $A=\CC^2$ and
\begin{itemize}
\item $\boldsymbol z^\spinup=\langle(1,0)\rangle$,
\item $\boldsymbol z^\spindown=\langle(0,1)\rangle$,
\item $\boldsymbol z=\CC^2$.
\end{itemize}
Measurements of spin along $x$ can be defined in a similar way in terms of the  basis of eigenvectors of the Pauli matrix $\sigma_x$, and we obtain
\[
\boldsymbol x=\boldsymbol x^\spinup\vee \boldsymbol x^\spindown=\CC^2=\boldsymbol z.
\]
Hence, in this representation, a fragment of the order of $\Max A$ that contains the above measurements is the following:
\[
\xymatrix{
&&\boldsymbol x=\boldsymbol z\ar@{-}[dll]\ar@{-}[dl]\ar@{-}[dr]\ar@{-}[drr]\\
\boldsymbol x^\spinup\ar@{-}[drr]&\boldsymbol x^\spindown\ar@{-}[dr]&&\boldsymbol z^\spinup\ar@{-}[dl]&\boldsymbol z^\spindown\ar@{-}[dll]\\
&&0
}
\]
\end{example}

The coincidence of $\boldsymbol x$ and $\boldsymbol z$ in the above example is undesirable, since the classical information obtained from these two measurements is clearly different (it corresponds to deflections in different planes), so they ought to be distinct measurements. A better representation is based on taking $A$ to be the matrix algebra $M_2(\CC)$:

\begin{example}\cite{fop}\label{exm:Cstarrep}
Consider again the measurements of spin along $x$ and $z$ as in the previous example, but now represented in $\Max M_2(\CC)$ by taking $\boldsymbol z^\spinup$, $\boldsymbol z^\spindown$, $\boldsymbol x^\spinup$, and $\boldsymbol x^\spindown$ to be spans of projection matrices:
\begin{align*}
&\boldsymbol z^{\spinup} = \left\langle\left(\begin{array}{cc}1&0\\0&0\end{array}\right)\right\rangle
\quad\quad \boldsymbol z^{\spindown} =  \left\langle\left(\begin{array}{cc}0&0\\0&1\end{array}\right)\right\rangle\\
&\boldsymbol x^{\spinup} = \left\langle\left(\begin{array}{cc}1&1\\1&1\end{array}\right)\right\rangle\quad\quad
\boldsymbol x^{\spindown} =  \left\langle\left(\begin{array}{cc}1&-1\\-1&1\end{array}\right)\right\rangle
\end{align*}
Now $\boldsymbol z$ and $\boldsymbol x$ are distinct abelian subalgebras of $M_2(\CC)$, generated by the Pauli spin matrices $\sigma_z$ and $\sigma_x$:
\begin{align*}
\boldsymbol z = D_2(\CC) &=\left\langle \left(\begin{array}{cc}1&0\\0&1\end{array}\right)\ ,\ \left(\begin{array}{cc}1&0\\0&-1\end{array}\right)\right\rangle\\
\boldsymbol x &=\left\langle \left(\begin{array}{cc}1&0\\0&1\end{array}\right)\ ,\ \left(\begin{array}{cc}0&1\\1&0\end{array}\right)\right\rangle.
\end{align*}
\end{example}

\paragraph{Sup-linear maps.}

Let $A$ be a locally convex space. A join $\langle a\rangle\vee \langle b\rangle=\langle a,b\rangle$ in $\Max A$ conveys a notion of disjunction of $a$ and $b$ that consists of closure under quantum superpositions, and $\Max A$ can be regarded as the result of freely adding such disjunctions to (the projectivization of) $A$. Let us make the latter statement precise in terms of a universal property which is based on the following notion:

\begin{definition}
Let $A$ be a vector space, and $L$ a sup-lattice. A map $f:A\to L$ is \emph{sup-linear} if $f(0)=0$ and $f(\lambda a+b)\le f(a)\vee f(b)$ for all $a,b\in A$ and $\lambda\in\CC$.
\end{definition}

\begin{lemma}
Let $A$ be a vector space, $L$ a sup-lattice and $f:A\to L$ a map. The following conditions are equivalent:
\begin{enumerate}
\item\label{propsuplin1} $f$ is sup-linear;
\item\label{propsuplin2} $\V f\bigl(\langle S\rangle\bigr)\le\V f(S)$ for all $S\subset A$;
\item\label{propsuplin3} $\V f\bigl(\langle S\rangle\bigr)=\V f(S)$ for all $S\subset A$.
\end{enumerate}
Moreover, if $f$ is sup-linear then for all linear subspaces $P,Q\subset A$ we have
\begin{equation}\label{propsuplin4}
\V f(P+Q) = \V f(P)\vee\V f(Q).
\end{equation}
\end{lemma}

\begin{proof}
\eqref{propsuplin2} and \eqref{propsuplin3} are obviously equivalent because $f(S)\subset f\bigl(\langle S\rangle\bigr)$. Let us assume \eqref{propsuplin1} and prove \eqref{propsuplin2}. Let $S\subset A$. If $S=\emptyset$ we have
$\langle S\rangle=\{0\}$ and thus \eqref{propsuplin2} is equivalent to
$f(0)\le 0$, which is true because we are assuming that $f$ is sup-linear. Now let $x\in \langle S\rangle$, say $x=\lambda_1 x_1+\cdots+\lambda_n x_n$ for some $x_1,\ldots,x_n\in S$ and $\lambda_1,\ldots,\lambda_n\in\CC$. Then sup-linearity gives us
\[
f(x)=f(\lambda_1 x_1+\cdots+\lambda_n x_n)\le f(x_1)\vee\cdots\vee f(x_n)\le \V f(S),
\]
so \eqref{propsuplin2} holds. Finally, \eqref{propsuplin4} is an immediate consequence of \eqref{propsuplin3}:
\begin{eqnarray*}
\V f(P+Q) &=& \V f\bigl(\langle P\cup Q\rangle\bigr)=\V f(P\cup Q)\\
&=&
\V f(P)\cup f(Q)=\V f(P)\vee\V f(Q). \qedhere
\end{eqnarray*}
\end{proof}

It is easy to see that for all locally convex spaces $A$ the span map
\[
\spanmap=\langle-\rangle:A\to\Max A
\]
is continuous (see~\cite{RS}), and it is obviously sup-linear. The following proposition describes the envisaged universal property:

\begin{theorem}\label{smallvsbigcont}
Let $A$ be a locally convex space. For each sober lattice $L$ and each continuous sup-linear map $f:A\to L$ there is a unique homomorphism of sober lattices $f^\sharp:\Max A\to L$ such that the following diagram commutes:
\[
\xymatrix{
A\ar[drr]_f\ar[rr]^-{\spanmap}&&\Max A\ar[d]^{f^\sharp}\\
&&L.
}
\]
Concretely, $f^\sharp$ is defined for each $P\in\Max A$ by
\begin{equation}\label{uniqueext}
f^\sharp(P)=\V_{p\in P} f(p).
\end{equation}
\end{theorem}

\begin{proof}
The uniqueness, as well as \eqref{uniqueext}, is a consequence of the fact that every $P\in\Max A$ is the join of the spans $\langle p\rangle$ for $p\in P$, together with the fact that we are requiring $f^\sharp$ to preserve joins:
\[
f^\sharp(P)=f^\sharp\bigl(\V_{p\in P}\langle p\rangle\bigr)
=\V_{p\in P} f^\sharp\bigl(\langle p\rangle\bigr)=\V_{p\in P} f(p).
\]
In addition, the sup-linearity of $f$ implies
\[
f^\sharp(\langle0\rangle)=f(0)=0,
\]
and for all $P,Q\in \Max A$ we have
\begin{eqnarray*}
f^\sharp(P\vee Q)&=&f^\sharp\bigl(\overline{P+Q}\bigr)=\V f\bigl(\overline{P+Q}\bigr)\\
&=&\V\overline{f\bigl(\overline{P+Q}\bigr)}\quad\textrm{(by Proposition~\ref{tinyprop})}\\
&=&\V\overline{f(P+Q)}\quad\textrm{(because $f$ is continuous)}\\
&=&\V f(P+Q)\quad\textrm{(by Proposition~\ref{tinyprop})}\\
&=&\V f(P)\vee\V f(Q)\quad\textrm{(by sup-linearity of $f$)}\\
&=&f^\sharp(P)\vee f^\sharp(Q).
\end{eqnarray*}
The only thing remaining to be proved is that $f^\sharp$ is continuous. Let $P\in\Max A$ and $U\subset L$ be an open set containing $f^\sharp(P)$. We shall find an open set $\mathcal U\subset\Max A$ containing $P$ such that $f^\sharp(\mathcal U)\subset U$. In order to do this, notice that the equation
\[
f^\sharp(P)=\V_{p\in P} f(p)
\]
means that $f^\sharp(P)$ is the value of the continuous join map $\V:L^P\to L$ at the family $\bigl(f(p)\bigr)_{p\in P}$. Then, by the continuity of $\V$, there is a family $(U_p)_{p\in P}$ of opens of $L$ such that the following conditions hold:
\begin{itemize}
\item all but finitely many $U_p$'s equal $L$;
\item $f(p)\in U_p$ for each $p\in P$;
\item for all $(x_p)_{x\in P}\in\prod_{p\in P} U_p$ we have $\V_{p\in P} x_p\in U$.
\end{itemize}
Let $F\subset P$ be a finite set such that $U_p=L$ for all $p\in P\setminus F$,  for each $p\in F$ let $V_p\subset A$ be the open $f^{-1}(U_p)$, and let
$\mathcal U$ be the finite intersection $\bigcap_{p\in F} \ps V_p$, which is an open set of $\Max A$. Since $p\in V_p\cap P$ for each $p$ it follows that $P\in\mathcal U$. Now let $Q\in\mathcal U$, for each $p\in F$ choose
$a_p\in Q\cap V_p$, and extend the finite family $(a_p)_{p\in F}$ in an arbitrary way to a family $(a_p)_{p\in P}$ in $Q$ --- for instance making $a_p=0$ for all $p\in P\setminus F$. Then $\bigl(f(a_p)\bigr)_{p\in P}\in\prod_{p\in P} U_p$, and thus
\[
f^\sharp(Q)=\V_{q\in Q} f(q)\ge\V_{p\in P} f(a_p)\in U,
\]
so $f^\sharp(Q)\in U$ because $U$ is upper-closed.
\end{proof}

From here on I shall denote by $\LCS$ the category of locally convex spaces and continuous linear maps between them.

\begin{corollary}\label{cor:soblatfunctor}
$\Max$ extends to a functor from $\LCS$ to $\SobL$ which to each continuous linear map $f:A\to B$ of locally convex spaces assigns the homomorphism of sober lattices $\Max f$ defined for all $P\in\Max A$ by
\[
\Max f(P)=\overline{f(P)}.
\]
\end{corollary}

\begin{proof}
For all continuous linear maps $f:A\to B$ the composition
\[
\xymatrix{A\ar[r]^f&B\ar[r]^-\spanmap&\Max B}
\]
is continuous and sup-linear, and we define $\Max f:\Max A\to\Max B$ to be the  extension $(\spanmap\circ f)^\sharp$. The functoriality follows from the uniqueness of the extension, and the formula for $\Max f$ is obtained from \eqref{uniqueext}.
\end{proof}

\subsection{Classical measurements}\label{sec:classical}

The distinction between classical and non-classical spaces of measurements does not have to do with disjunctions alone, but rather with how they interact with conjunctions, and with computational considerations like those underlying the notion of domain in computer science. This leads to the conclusion that spaces of classical measurements are countably based locally compact locales --- \ie, up to isomorphism they are topologies of second-countable locally compact spaces --- equipped with the Scott topology. The main arguments behind this were given in~\cite{fop}, and here I shall provide complementary details.

\paragraph{Conjunctions and distributivity.}

If $L$ is a sober lattice then, because it is a complete lattice, it has meets of any subsets, in particular a binary meet $x\wedge y$ for each $x,y\in L$. Bearing in mind the interpretation of $L$ as a space of measurements, the binary meet $x\wedge y$ is a measurement which is \emph{both} $x$ \emph{and} $y$. This is not in the sense of having both the observable properties of $x$ and those of $y$ (that would be the disjunction $x\vee y$), but rather that every observable property of $x\wedge y$ is both an observable property of $x$ and an observable property of $y$. We say that $x\wedge y$ is the \emph{conjunction} of $x$ and $y$.

\begin{example}
Let $\boldsymbol z^\spindown$ and $\boldsymbol z^\spinup$ be the spin measurements of Example~\ref{exm:Cstarrep}. The meet $\boldsymbol z^\spindown\wedge \boldsymbol z^\spinup$ corresponds to a conjunction of incompatible alternatives, namely that a deflection was observed which is simultaneously upwards and downwards. So the conjunction of $\boldsymbol z^\spindown$ and $\boldsymbol z^\spinup$ should, as indeed is the case in $\Max M_2(\CC)$, coincide with the impossible measurement: $\boldsymbol z^\spindown\wedge\boldsymbol z^\spinup=0$.
\end{example}

Sober lattices in general are not distributive: given three measurements $m$, $n$ and $p$, only the inequality
\begin{equation}\label{distrineq}
(p\wedge m)\vee(p\wedge n)\le p\wedge(m\vee n)
\end{equation}
is guaranteed to hold. Similarly to the quantum logic of Birkhoff and von Neumann, lack of distributivity is a sign of nonclassical behavior. Intuitively, measurements of classical type should be those that allow us to pretend that we are measuring intrinsic, and mutually independent (\ie, non-interfering), properties of systems, such that a measurement $p\wedge (m\vee n)$, which reads \emph{both $p$ and the disjunction $m\vee n$}, cannot be distinguished from the disjunction \emph{(both $p$ and $m$) or (both $p$ and $n$)}~\cite{fop}. Accordingly, a situation where distributivity always exists is the following:

\begin{lemma}\label{lem:composedsys}
Let $L$ and $M$ be sober lattices. For all $x\in L$ and all $m,n\in M$ the following condition holds in $L\oplus M$:
\[
(x,1)\wedge\bigl((1,m)\vee(1,n)\bigr)
= \bigl((x,1)\wedge (1,m)\bigr)\vee\bigl((x,1)\wedge (1,n)\bigr).
\]
\end{lemma}

\begin{proof}
Similarly to joins, the binary meets are computed pairwise, so we have
\begin{eqnarray*}
(x,1)\wedge\bigl((1,m)\vee(1,n)\bigr)&=&(x,1)\wedge (1,m\vee n)\\
&=&(x,m\vee n)\\
&=&(x\vee x,m\vee n)\\
&=&(x,m)\vee(x,n)\\
&=&\bigl((x,1)\wedge(1,m)\bigr)\vee\bigl((x,1)\wedge(1,n)\bigr). \qedhere
\end{eqnarray*}
\end{proof}

An interpretation of this lemma is that a direct sum $L\oplus M$ can represent the measurements on a system composed of two non-interfering parts whose measurements are respectively described by $L$ and $M$, such that each pair $(x,m)$ represents a conjunction of $x$ and $m$, whereas doing only $x$ is interpreted as doing $x$ on the system represented by $L$ and doing the completely undetermined measurement $1$ on the system represented by $M$ (that is, anything at all may have happened to $M$ while $x$ was performed). So $x$ becomes the pair $(x,1)$ and, similarly, doing $m\in M$ alone is represented by the pair $(1,m)$. As expected for a conjunction, $(x,m)$ is a meet: $(x,m)=(x,1)\wedge (1,m)$.

\begin{example}
Referring again to a spin measurement of an electron, let there be two possible  intensities for the gradient of the magnetic field along $z$, due to which there are two possible magnitudes, $m_1$ and $m_2$, of the deflection in the Stern--Gerlach apparatus. Then $m_1\wedge \boldsymbol z^\spindown$ is a measurement of spin along $z$ which besides recording a downwards deflection also determines its magnitude to be $m_1$. In this case it makes sense to have distributivity:
\[
m_1\wedge \boldsymbol z = (m_1\wedge \boldsymbol z^\spindown) \vee (m_1\wedge \boldsymbol z^\spinup).
\]
An example of a sober lattice that models such measurements is $\boldsymbol 2^2\oplus\Max M_2(\CC)$ (\cf\ Example~\ref{exm:22}). This is the product of a classical bit, which describes the two possible magnitudes, with the space of spin measurements.
\end{example}

\paragraph{Locally compact locales and classical state spaces.}

The same argument used in computer science~\cite{Stoy}*{chapter 6} when modeling spaces of data values and computable functions on them applies to spaces of measurements in those situations where new information about a system never contradicts previously obtained information, so that each new measurement can be regarded simply as more data about the system~\cite{fop}. Mathematically, this implies that spaces of classical measurements must be countably-based continuous lattices, and they can be topologized with the Scott topology itself. Moreover, for any continuous lattice $L$ equipped with the Scott topology the product topology on $L\times L$ coincides with the Scott topology~\cite{bigdomainbook}*{Th.\ II-1.14 and Th.\ II-4.13}, and thus any function $\theta:L\times L\to L$ in two variables on $L$ is continuous with respect to the product topology if and only if it is continuous in each variable separately~\cite{bigdomainbook}*{Lemma II-2.8}. Hence, we obtain:

\begin{lemma}\label{lem:contscottsober}
Any continuous lattice equipped with the Scott topology is a sober lattice.
\end{lemma}

\begin{proof}
A function $f$ is continuous for the Scott topology if and only if for all directed sets $D$ in the domain of $f$ we have~\cite{bigdomainbook}*{Prop.\ II-2.1}
\[
f(\V D) = \V f(D).
\]
Hence, since directed sets are nonempty, if a continuous lattice $L$ is equipped with the Scott topology the operation $\vee:L\times L\to L$ is continuous in each variable separately, and thus it is continuous. Moreover, it is a property of continuous lattices that the Scott topology makes them sober spaces~\cite{bigdomainbook}*{Th.\ II-1.14 and Cor.\ II-4.16}.
\end{proof}

It is well known that any continuous lattice which is also distributive is a \emph{locale}, \ie, a sup-lattice $L$ that satisfies the following infinite distributivity property for all $x\in L$ and $S\subset L$:
\[
x\wedge \V S = \V_{y\in S} x\wedge y.
\]
Moreover, any continuous locale $L$ is necessarily order isomorphic to the topology $\topology(X)$ of a sober locally compact space $X$ (locally compact in the strong sense that for any $x\in X$ and any open neighborhood $U$ of $x$ there is an open neighborhood $V$ of $x$ and a compact set $K$ such that $x\in V\subset K\subset U)$, and, conversely, the topology of a sober locally compact space is necessarily a continuous locale.
For this reason, distributive continuous lattices are usually called \emph{locally compact locales}. The condition that $\topology(X)$ is countably-based means precisely that $X$ is second-countable.

\begin{definition}\label{def:classicalmeasurements}
A sober lattice is referred to as a \emph{space of classical measurements} if it is a countably-based locally compact locale equipped with the Scott topology.
\end{definition}

A consequence of this definition is that there is a natural notion of \emph{topological spectrum} associated to classical measurements: given a locally compact locale $L$, any sober locally compact space $X$ such that $L\cong\topology(X)$ is unique up to a homeomorphism, and it is homeomorphic to the locale spectrum of $L$~\cite{stonespaces}.

\begin{corollary}
A sober lattice is a space of classical measurements if and only if it is homeomorphic to the topology $\topology(X)$, itself equipped with the Scott topology, of a second-countable sober locally compact space $X$. Moreover, any space of classical measurements is separable.
\end{corollary}

\begin{proof}
Only the statement about separability needs some explanation. Let $X$ be a locally compact space with a countable basis $\mathcal B$, and let $\mathcal U$ be a non-empty open set of $\topology(X)$ in the Scott topology. Then there is $U\in\topology(X)$ such that $U\in\mathcal U$. Express $U$ as a union of open sets taken from $\mathcal B$:
\[
U = \bigcup_{i\in I} B_i.
\]
Then for some finite subset $F\subset I$ we have $\bigcup_{i\in F} B_i\in\mathcal U$. So we see that the set $\check {\mathcal B}$ (the completion of $\mathcal B$ under finite unions) is again countable, and it is dense.
\end{proof}

\paragraph{Example.}

Let us briefly address the two-slit experiment,
modeling the experimental apparatus as a space of classical measurements. Take the target where interference patterns are recorded to be a copy of the real line. So we can take $\topology(\RR)$ to be the space of position measurements in the target. Secondly, as usual we may or may not be able to detect which slit the particle went through. Recording that it went through the first slit will be the measurement $\slit_1$, and recording that it went through the second slit will be the measurement $\slit_2$. Stating that it must have gone through either one slit or the other gives us the equation
\[
\slit_1\vee\slit_2=1,
\]
and the impossibility of simultaneously recording $\slit_1$ and $\slit_2$ translates to
\[
\slit_1\wedge\slit_2=0,
\]
so we obtain the following Boolean algebra $S\cong\boldsymbol 2^2$:
\[
\xymatrix{
&1\ar@{-}[dl]\ar@{-}[dr]\\
\slit_1\ar@{-}[dr]&&\slit_2\ar@{-}[dl]\\
&0
}
\]
The classical measurement space that describes the whole apparatus is 
\[
M=S\oplus\topology(\RR)\cong\topology\bigl(\{\slit_1,\slit_2\}\amalg\RR\bigr)
\]
(\cf\ Lemma~\ref{lem:composedsys} and the comments that follow).
For instance, if $U$ is an open set of $\RR$, the pair $(\slit_1,U)$ is a measurement that consists of detecting both that the particle traversed the first slit and that it landed in the target region described by $U$, and we have
\[
(\slit_1,U)\vee(\slit_2,U)=(1,U).
\]

This description abstracts away the experimental details of how the detection of particles through slits is done (and also ignores the subtle effects produced by the change in boundary conditions that occurs when, say, we block one of the slits). For instance, it is irrelevant whether we are running a standard two-slit experiment or Wheeler's delayed choice experiment: if the particle is ultimately known to have gone through the first slit and to have landed in $U$, then the abstract experiment being performed is $(\slit_1,U)$, whereas if it has been detected through the second slit the abstract experiment is $(\slit_2,U)$. But, if no detection through any of the slits has been made, the experiment is the disjunction of the previous two, in other words $(1,U)$.

Summarizing: \emph{the only actual two-slit experiment that records a hit of the particle in $U$ is $(1,U)$, whereas $(\slit_1,U)$ and $(\slit_2,U)$ are ``one-slit experiments.''} This classification of measurements is exactly that which is afforded by their classical information content: if a particle is detected through one of the slits, no matter how this was done, then the measurement is, in terms of the classical information it elicits, not a mysterious ``two-slit experiment with some subtle modification about one particular slit,'' but, rather, it is a one-slit experiment.

A full treatment of this example, accounting for interference patterns on the target, needs to be left for a continuation of this paper that addresses measure-theoretic structure on measurement spaces.

\section{Sequential compositions and reversibility}\label{sec:MS}

In this section I address the algebraic structure of measurements that relates to time and causality. Although there are several options for algebraic operations that represent causality structures~\cite{fop}, here I shall focus only on the basic sequential composition of measurements that underlies the definition of measurement space. This will encompass spaces of classical measurements, in particular a broad class of examples based on locally compact groupoids, and spaces of quantum measurements based on C*-algebras.

\subsection{Measurement spaces and quantales}\label{sec:quantales}

\paragraph{Measurement spaces.}

The meaning of the continuous multiplication of measurements introduced in~\cite{fop} is that a product $mn$ is a \emph{composition} of measurements read from right to left: $n$ \emph{and then} $m$. For instance, an electron may traverse two Stern--Gerlach analyzers one after the other, as in Fig.~\ref{fig:zzup}, where the composition $\boldsymbol z\boldsymbol z^\spinup$ conveys exactly the same information as $\boldsymbol z^\spinup$, so the two measurements should be considered the same: $\boldsymbol z\boldsymbol z^\spinup=\boldsymbol z^\spinup$.
\begin{figure}[h!]
\begin{center}\includegraphics[width=0.8\textwidth]{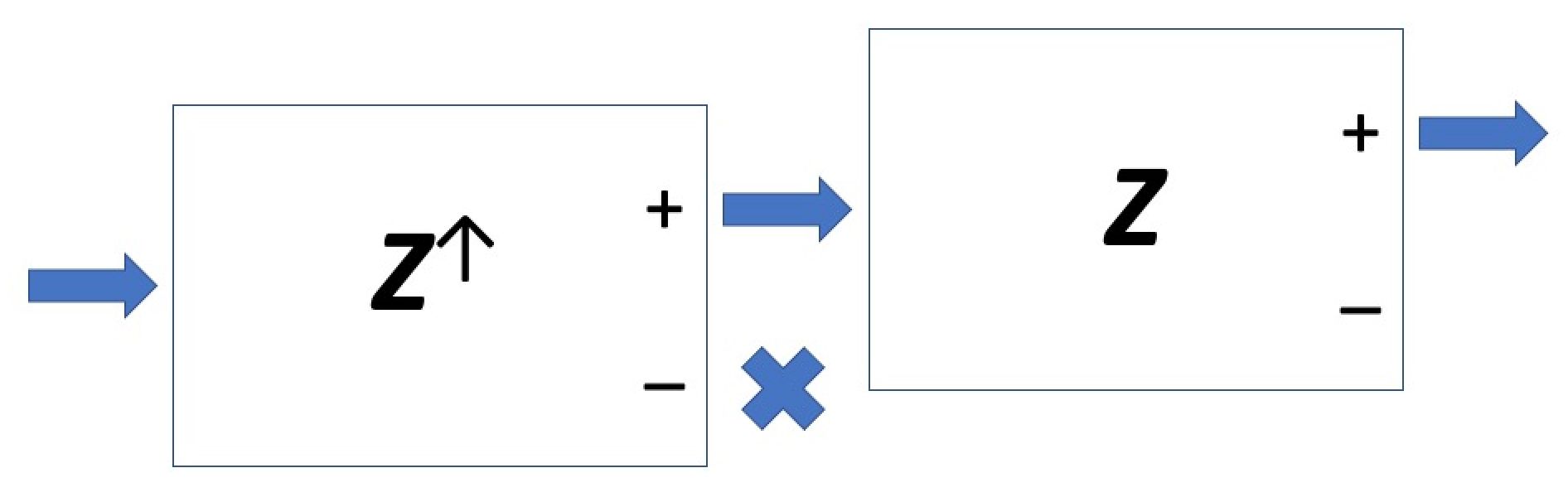}\caption{Sequential composition ``$\boldsymbol z^\spinup$ and then $\boldsymbol z$.'' After $\boldsymbol z^\spinup$ the measurement $\boldsymbol z$ gives no new information.}\label{fig:zzup}\end{center}
\end{figure}

Besides the multiplication, measurement spaces are equipped with a continuous involution operation $m\mapsto m^*$. This is a purely formal time reversal which allows one to express a necessary condition for a measurement $m$ to be considered reversible, namely $mm^*m=m$.

\begin{definition}\label{def:ms}\cite{fop}
By a \emph{measurement space} is meant a sober lattice $M$ equipped with a continuous associative \emph{product} $(m,n)\mapsto mn$ and a continuous \emph{involution} $m\mapsto m^*$ that satisfy the following properties for all $m,n,n'\in M$:
\begin{enumerate}
\item $m0=0$ (absorption);
\item $m(n\vee n')=mn\vee mn'$ (distributivity);
\item $m^{**}=m$ (involution I);
\item $(mn)^*=n^*m^*$ (involution II);
\item $mm^*m\le m\Longrightarrow m\le mm^*m$ (stably Gelfand).
\end{enumerate}
A \emph{homomorphism} of measurement spaces $h:M\to N$ is a homomorphism of sober lattices such that for all $m,m'\in M$ we have
\[
h(mm')=h(m)h(m')\quad\textrm{and}\quad h(m^*)=h(m)^*.
\]
The ensuing \emph{category of measurement spaces} is denoted by $\MS$. If the product has a unit $e$ for the multiplication the measurement space is said to be \emph{unital}.
\end{definition}

\begin{lemma}
In any measurement space $M$ the following conditions hold for all $m,n,p\in M$ and all $S\subset M$:
\begin{enumerate}
\item\label{msp1} $(-)^*:M\to M$ is a homeomorphism;
\item\label{msp2}  $(-)^*:M\to M$ is an order isomorphism;
\item\label{msp3}  $(\V S)^*=\V_{n\in S} n^*$;
\item\label{msp4}  $m(\V S)=\V_{n\in S}mn$;
\item\label{msp5}  $(\V S)m=\V_{n\in S}nm$;
\item\label{msp5.1} $(m\vee n)p=mp\vee np$;
\item\label{msp6}  $0m=0$;
\item\label{msp7}  $mm^*m\le m\iff mm^*m=m$;
\item\label{msp8}  $e^*=e$ if $M$ is unital.
\end{enumerate}
Moreover, if $h:M\to N$ is a homomorphism of measurement spaces then for all $S\subset M$ we have
\begin{enumerate}
\setcounter{enumi}{9}
\item\label{msp9} $h\bigl(\V S\bigr) = \V f(S)$.
\end{enumerate}
\end{lemma}

\begin{proof}
\eqref{msp1} The continuous involution is its own inverse. \eqref{msp2} Homeomorphisms are order isomorphisms of the specialization order. \eqref{msp3} Order isomorphisms preserve all the joins. \eqref{msp4} Continuity of the product implies that it preserves joins of directed sets in the right variable, so due to distributivity and absorption all the joins are preserved in that variable. \eqref{msp5}--\eqref{msp6} Apply the involution to \eqref{msp4}. \eqref{msp7} Obvious. \eqref{msp8} $e^*=e^*e=e^* e^{**}=(e^* e)^*=e^{**}=e$. Finally, \eqref{msp9} is a consequence of the continuity of $h$, which implies that it preserves directed joins, and the fact that it preserves joins of finite subsets.
\end{proof}

\paragraph{Quantales.}

The above lemma shows that, algebraically, measurement spa\-ces are stably Gelfand quantales~\cite{GSQS,SGQ}. Let us recall the main definitions surrounding the notion of quantale:

\begin{definition}
By a \emph{quantale} is meant a sup-lattice $Q$ equipped with an associative multiplication $(a,b)\mapsto ab$ that distributes over arbitrary joins in each variable:
\[
a\bigl(\V_i b_i\bigr) = \V_i(ab_i)\quad\textrm{and}\quad\bigl(\V_i a_i\bigr) b=\V_i(a_i b).
\]
A quantale is \emph{unital} if there is unit for the multiplication, usually denoted by $e$.

Given a quantale $Q$, an element $a\in Q$ is \emph{right-sided} if $a1\le a$,  \emph{left-sided} if $1a\le a$, and \emph{two-sided} if it is both right- and left-sided. The set of right-sided elements of $Q$ is denoted by $\rs(Q)$, that of left-sided elements is denoted by $\ls(Q)$, and the set of two-sided elements is denoted by $\ts(Q)$.

An \emph{involutive quantale} is a quantale equipped with a semigroup involution $a\mapsto a^*$ that preserves arbitrary joins:
\[
\bigl(\V_i a_i\bigr)^*=\V_i a_i^*.
\]
An involutive quantale $Q$ is \emph{stably Gelfand} if for all $a\in Q$ we have
\[
aa^*a\le a\ \Longrightarrow\  a\le aa^*a,
\]
and it is \emph{Gelfand} (resp.\ \emph{strongly Gelfand}) if for all $a\in \rs(Q)$ (resp.\ all $a\in Q$) we have
\[
a\le aa^*a.
\]

A \emph{homomorphism} of quantales is a mapping which is both a sup-lattice homomorphism and a semigroup homomorphism. A homomorphism between unital quantales is \emph{unital} if it is also a monoid homomorphism, and a homomorphism between involutive quantales is \emph{involutive} if it commutes with the involutions.

By a \emph{subquantale} of $Q$ is meant a subset $S\subset Q$ closed under joins and the multiplication. Unital subquantales and involutive subquantales have the obvious meanings.
\end{definition}

The following are well known simple facts:

\begin{proposition}
Let $Q$ be a quantale.
\begin{enumerate}
\item We have the following lattice of subquantale inclusions:
\[
\xymatrix{
&Q\ar@{-}[dl]\ar@{-}[dr]\\
\ls(Q)\ar@{-}[dr]&&\rs(Q)\ar@{-}[dl]\\
&\ts(Q).
}
\]
\item If $Q$ is involutive, the involution restricts to an order isomorphism
\[
\rs(Q)\to\ls(Q).
\]
\item Any strongly Gelfand quantale is stably Gelfand, and any stably Gelfand quantale is Gelfand.
\end{enumerate}
\end{proposition}

\begin{example}\label{exm:quantales}
Let us list some examples of quantales:
\begin{enumerate}
\item Any locale is a unital strongly Gelfand quantale for the multiplication given by binary meet, $ab=a\wedge b$, and the trivial involution $a^*=a$. The unit $e$ coincides with the top element $1$. Conversely, any unital quantale whose multiplication is idempotent and whose unit coincides with the top is a locale~\cite{JT}.
\item If $Q$ is any Gelfand quantale then $\ts(Q)$ is a locale, and its involutive quantale structure as a subquantale of $Q$ coincides with the locale structure as in the previous example~\cite{MP1}.
\item Given any set $X$, the set $\mathcal P(X\times X)$ of binary relations on $X$ under the inclusion order is a unital strongly Gelfand quantale whose multiplication is given by composition of binary relations,
\[
RS = R\circ S = \{(x,y)\st \exists_{z\in X}\ (x,z)\in R\text{ and }(z,y)\in S\},\\
\]
whose involution is relation reversal,
\[
R^* = \{(x,y)\st (y,x)\in R\},
\]
and whose unit is the diagonal relation $e=\Delta_X$.

\item More generally, in sections~\ref{sec:groupoids} and \ref{sec:etale} it will be seen that the topology of any open topological groupoid $G$ is a strongly Gelfand quantale $\opens(G)$~\cite{Re07}, unital if $G$ is \'etale; and
we have order isomorphisms $\rs(\opens(G))\cong\ls(\opens(G))\cong\topology(G_0)$, where $G_0$ is the object space of $G$.

\item Given any ring $R$, the set $\Sub(R)$ of additive subgroups of $R$ is a quantale, unital if $R$ has a unit, in which case $e$ is the additive cyclic group $\langle 1_R\rangle$. The right-sided elements of $\Sub(R)$ are precisely the right ideals of $R$, and the left-sided elements are the left ideals. See~\cite{Rosenthal1}. If $R$ is an involutive ring then $\Sub(R)$ is an involutive quantale with the pointwise involution. Analogous facts hold for topological rings and closed additive subgroups, closed right ideals, etc., or, more generally, for arbitrary topological $F$-algebras for some field $F$. Section~\ref{sec:Cstar} will address this example with $F=\CC$ for the specific case of C*-algebras.
\end{enumerate}
\end{example}

\begin{remark}
An equivalent definition of measurement space is that it is a stably Gelfand quantale equipped with a sober topology with respect to which the multiplication, the involution, and the joins are continuous; and a homomorphism of measurement spaces is the same as a continuous homomorphism of involutive quantales. However, Definition~\ref{def:ms} imposes only finitary algebraic structure, contrasting with the usual second-order definitions of quantale and quantale homomorphism, and it provides a distinction between ``essential'' algebraic structure and that which is derived from the topology. This is useful from the point of view of the ontological interpretation of measurement spaces.
\end{remark}

\subsection{C*-algebras}\label{sec:Cstar}

Algebraic quantum theory is based on describing quantum systems by means of C*-algebras that contain their observables, and therefore it is especially relevant to study measurement spaces associated to C*-algebras. This has already been done in the spin example (\cf\ Example~\ref{exm:Cstarrep}) in terms of the algebra $M_2(\CC)$, albeit without taking multiplication and involution into account.

\begin{theorem}
Let $A$ be a C*-algebra. Then $\Max A$ is a measurement space, separable if $A$ is separable, whose involution is computed pointwise and whose multiplication is defined for all $P,Q\in\Max A$ by
\[
PQ = \overline{\langle P\cdot Q\rangle}
\]
where $P\cdot Q$ is the pointwise product $\{ab\st a\in P,\ b\in Q\}$.
If $B$ is another C*-algebra and $f:A\to B$ is a $*$-homomorphism the mapping
\[
\Max f:\Max A\to\Max B
\]
defined by
$\Max f(P)=\overline{f(P)}$
is a homomorphism of measurement spaces.
\end{theorem}

\begin{proof}
We have already seen that $\Max A$ is a sober lattice, separable if $A$ is. 
The involutive quantale structure of $\Max A$ is well known from~\cite{MP1}, where the functoriality of $\Max$ was established along with the formula for $\Max f$. Moreover, Corollary~\ref{cor:soblatfunctor} shows that $\Max f$ is continuous, and it is also known that $\Max A$ is stably Gelfand~\cite{QFB}. Hence, in order to show that $\Max A$ is a measurement space all we need to do is prove that the multiplication and the involution are continuous with respect to the lower Vietoris topology. Let us begin with the involution. The involution of $A$ is continuous, so Corollary~\ref{cor:soblatfunctor} shows that
it yields a continuous map $\Max A\to\Max A$ given for all $P\in\Max A$ by
\[
P\mapsto\overline{P^*}
\]
where $P^*$ is the pointwise involute of $P$. Moreover, the involution of $A$ is a homeomorphism, so $P^*$ is closed, and thus the mapping $P\mapsto P^*$, which is the involution of $\Max A$, is continuous.

Now let us prove that the multiplication of $\Max A$ is continuous.
Let $P,Q\in\Max A$, and suppose $\mathcal U\in\topology(\Max A)$ is such that $PQ\in\mathcal U$. I will show that there exist open sets $\mathcal V,\mathcal W\in\topology(\Max A)$ such that $P\in\mathcal V$, $Q\in\mathcal W$, and $\mathcal V\mathcal W\subset\mathcal U$. It suffices to do this for $\mathcal U=\ps U$ with $U\in\topology(A)$, in which case the condition $PQ\in\ps U$ translates to $\langle P\cdot Q\rangle\cap U\neq\emptyset$. Let then $a_1,\ldots, a_k\in P$ and $b_1,\ldots,b_k\in Q$ ($k\ge 1$) be such that
\[
a_1 b_1+\cdots+a_k b_k\in U.
\]
Due to the continuity of the multiplication of $A$ there are
\[
V_1,\ldots,V_k,W_1,\ldots, W_k\in\topology(A)
\]
such that $V_1 W_1+\cdots + V_k W_k\subset U$ and $a_i\in V_i$ and $b_i\in W_i$ for all $i=1,\ldots,k$. Then making $\mathcal V=\ps V_1\cap\ldots\cap \ps V_k$ and $\mathcal W=\ps W_1\cap\ldots\cap \ps W_k$ we obtain $P\in\mathcal V$ and $Q\in\mathcal W$. Now let $P'\in\mathcal V$ and $Q'\in\mathcal W$, and choose elements $a_i'\in P'\cap V_i$ and $b_i'\in Q'\cap W_i$ for each $i=1,\ldots, k$. Then
\[
a'_1 b'_1+\cdots+a'_k b'_k\in P'Q'\cap U,
\]
so $P'Q'\in\ps U$. This shows that $\mathcal V\mathcal W\subset\ps U$, so the multiplication of $\Max A$ is continuous.
\end{proof}

Let us denote the category of C*-algebras and $*$-homomorphisms by $\Cstarcat$.

\begin{corollary}
$\Max$ is a functor from $\Cstarcat$ to $\MS$.
\end{corollary}

Remarkably, at least for unital C*-algebras, the functor $\Max$ is a complete invariant, due to the following result (which considers only discrete quantales):

\begin{theorem}[\cite{KR}]\label{completeinv}
Let $A$ and $B$ be unital C*-algebras. There is a unital $*$-isomorphism $\varphi:A\to B$ if and only if there exists an isomorphism of unital involutive quantales $\alpha:\Max A\to\Max B$.
\end{theorem}

This raises the question of how much information the measurement space $\Max A$ contains about the C*-algebra $A$, and therefore also about any quantum system which is described by $A$. While on one hand $\Max A$ classifies $A$ up to $*$-isomorphisms, on the other it does not classify the isomorphisms themselves, which means there is loss of structural information in the assignment $A\mapsto\Max A$. In particular, in general
\[
\Aut(A)\not\cong\Aut(\Max A),
\]
and a simple example of this is obtained by taking $A=\CC^2$ because the involutive quantale $\Max A$ has the following automorphism $\alpha$ which is not of the form $\Max\varphi$~\cite{KR}:
\begin{equation}\label{eq:alpha}
\begin{array}{rcll}
\alpha\bigl( (z,w)\CC\bigr) &=&  (w,z)\CC &\textrm{if }z\neq 0\textrm{ and }w\neq 0\\
\alpha\bigl( (z,0)\CC\bigr) &=&  (z,0)\CC\\
\alpha\bigl( (0,w)\CC\bigr) &=& (0,w)\CC.
\end{array}
\end{equation}
However, the lower Vietoris topology rules out this pathology, thus showing that measurement spaces provide a better complete invariant of unital C*-algebras than discrete involutive quantales do:

\begin{theorem}
The automorphism $\alpha\in\Aut(\Max \CC^2)$ of \eqref{eq:alpha} is not continuous with respect to the lower Vietoris topology of $\Max\CC^2$.
\end{theorem}

\begin{proof}
Let us assume that $\Max \alpha$ is continuous and derive a contradiction. We have $\Max \alpha\bigl((1,\frac 1 n)\CC\bigr)=(\frac 1 n, 1)\CC$ for all $n\in\NN_{>0}$, and therefore the sequence $(\frac 1 n,1)\CC$ in $\Max \CC^2$ must converge to $\Max\alpha\bigl((1,0)\CC\bigr)=(1,0)\CC$. Let $D(0)$ and $D(1)$ be the open disks in $\CC$ with radius $1/2$ centered in $0$ and $1$, respectively, and let $U$ be the open set $D(1)\times D(0)\subset \CC^2$. Then $(1,0)\CC\in\ps U$ because $(1,0)\in U\cap (1,0)\CC$, whereas for each $n\in\NN_{>0}$ the condition $(\frac 1 n,1)\CC\in\ps U$ means that there is $\lambda_n\in\CC$ such that $\lambda_n(\frac 1 n,1)\in U$, so we have both
\[
\left\vert\frac{\lambda_n}n-1\right\vert<\frac 1 2\quad\textrm{and}\quad\vert\lambda_n\vert<\frac 1 2.
\]
This is impossible, so $\Max\alpha$ is not continuous.
\end{proof}

An in-depth study of the properties of the functor $\Max$ is beyond the scope of this paper. I shall only record the following useful lemma, which does not follow from Theorem~\ref{completeinv} and also applies to non-unital C*-algebras:

\begin{lemma}
Let $A$ be a C*-algebra. Then $A$ is commutative if and only if $\Max A$ is commutative.
\end{lemma}

\begin{proof}
It is trivial that if $A$ is commutative then so is $\Max A$. For the converse assume that $A$ is not commutative. Then there is an irreducible representation $\pi:A\to B(H)$ on a Hilbert space with dimension greater than $1$. Due to the transitivity theorem there is a copy of $M_2(\CC)$ in the image $\pi(A)$, so we can find $a,b\in A$ and $x\in H$ such that $\pi(ab)(x)=0$ and $\pi(ba)(x)\neq 0$. Then the induced principal left $\Max A$-module $P(H)$ (\cf~\cite{KR}) satisfies $\langle a\rangle\langle b\rangle\cdot\langle x\rangle=\{0\}$ and $\langle b\rangle\langle a\rangle\cdot\langle x\rangle\neq\{0\}$, so $\langle a\rangle \langle b\rangle\neq\langle b\rangle\langle a\rangle$, showing that $\Max A$ is not commutative.
\end{proof}

\subsection{Classical measurements and groupoids}\label{sec:groupoids}

\paragraph{Classical measurement spaces.}

Adopting the stance of section~\ref{sec:classical}, according to which classical measurements are those that are compatible with the idea that the system being measured has an intrinsic state that does not depend on observation, one is led to the following definition:

\begin{definition}\cite{fop}\label{def:classicalmeasurementspaces}
A measurement space $M$ is said to be \emph{classical} if, as a sober lattice, it is a space of classical measurements in the sense of Definition~\ref{def:classicalmeasurements}, and moreover it satisfies the following condition for all $m\in M$:
\[
m\le mm^*m.
\]
\end{definition}

In particular, the involutive quantale associated to a classical measurement space is strongly Gelfand. The intuitive justification for this rests on the ideas that the product $m^*m$ corresponds to performing $m$ and then undoing it to the extent possible, so that, even if the system is not necessarily brought back to its original state, at least the original state should be one of the possible outcomes~\cite{fop}.

\begin{example}\label{exm:relationalms}
Let $X$ be any countable set. The quantale of binary relations $\mathcal P(X\times X)$ is strongly Gelfand (\cf\ Example~\ref{exm:quantales}).
Therefore, since $X\times X$ with the discrete topology is a second countable locally compact sober space, its topology $\topology(X\times X)=\mathcal P(X\times X)$, carrying the Scott topology, is a classical measurement space.
This is a particular example of the measurement spaces associated to groupoids which will be addressed below. Here $X$ can be regarded as being the set of \emph{states} of a classical system whose changes of state are described by the binary relations.
\end{example}

\paragraph{Local measurement spaces.} A particular type of a classical measurement space is obtained from any countably based locally compact locale without adding any extra structure for the multiplication and the involution:

\begin{definition}\cite{fop}
By a \emph{local measurement space} is meant a classical measurement space $M$ for which the multiplication and the involution are given for all $m,n\in M$ by
\[
mn = m\wedge n\quad\text{and}\quad m^*=m.
\]
\end{definition}

Note that in any local measurement space $M$ the multiplication is commutative and idempotent, thus conveying that the order in which measurements are performed is irrelevant, and that performing the same measurement more than once adds no new information.
Moreover, in such a measurement space we can proceed by successive approximations, for if $m\le n$ we have $mn=m$, which can be interpreted as saying that performing $n$ can be improved by subsequently performing the ``sharper'' measurement $m$. Similarly, $nm=m$ implies that after $m$ has been observed no new information is obtained by performing the less precise measurement $n$. In particular, the least precise measurement $1$ is such that $m1=1m=m$ for all $m\in M$ (so $M$ is unital with $e=1$), so it adds nothing to any other measurement. More generally, for arbitrary $m,n\in M$ we have $mn\le m$ and $mn\le n$, meaning that performing both $m$ and $n$ provides information which is at least as sharp as that of either $m$ or $n$ alone.

Since every local measurement space is homeomorphic to one of the form $\topology(X)$ for some second-countable locally compact sober space $X$, 
we may ask what the role played by the points of $X$ is. If $X$ happens to be a discrete set then each point $x\in X$ can be identified with the singleton open set $\{x\}$. This is the case, for instance, when $X$ is just a finite set of values that an observable can have: we can identify measurements of those values with the values themselves. However, if $X$ is not discrete then this identification of points with singleton open sets disappears, but in any case we may regard $X$ as providing the set of \emph{states} of the system being measured.

\begin{example}\label{exm:labwall}\cite{fop}*{Example 7} Let $W=[a,b]\times[c,d]\subset \RR^2$ ($a<b$ and $c<d$) represent the surface of a \emph{wall} of the lab where we are making measurements. If we observe the wall by visual inspection, photons reflected by the wall carry information about regions of the wall rather than about infinitely precise points, so such visual measurements can be represented by open sets of $W$.
However, the space $W$ does provide a useful model of the local measurement space $\topology(W)$. And, since $W$ is sober (because it is Hausdorff), the set $W$ together with its topology can be recovered from the locale structure of $\topology(W)$ up to a homeomorphism. Hence, the information carried by the local measurement space $\topology(W)$ is equivalent to that of the space $W$, whose points are the ``states'' of the lab wall. But it should be stressed that the points are farther away from direct observation than the open sets are.
\end{example}

\paragraph{Groupoids.}

Groupoids provide a way of describing symmetries of a non-global type~\cite{Weinstein}, and they also yield examples of classical measurement spaces. Let us give a first introduction to this, essentially recalling basic definitions mostly to fix terminology and notation.

\begin{definition}\label{def:groupoid}
A \emph{topological group\-oid} $G$, or simply a \emph{group\-oid} since no other kind will be addressed here, consists of topological spaces $G_0$ (the space of \emph{objects}, or \emph{units}) and $G_1$ (the space of \emph{arrows}), together with continuous structure maps
\[G\ \ \ \ =\ \ \ \ \xymatrix{
G_2\ar[rr]^-m&&G_1\ar@(ru,lu)[]_i\ar@<1.2ex>[rr]^r\ar@<-1.2ex>[rr]_d&&G_0\ar[ll]|u
}\]
that satisfy the axioms listed below,
where: $G_2$ is the pullback of the \emph{domain map} $d$ and the \emph{range map} $r$,
\[
G_2=\{(g,h)\in G_1\times G_1\st d(g)=r(h)\};
\]
the map $m$ is the \emph{multiplication} and $i$ is the \emph{inversion}; the map $u$ is required to be a section of both $d$ and $r$, so
often $G_0$ is identified with a subspace of $G_1$ (with $u$ being the inclusion map). With this convention, and writing $gh$ instead of $m(g,h)$, as well as $g^{-1}$ instead of $i(g)$, the axioms satisfied by the structure maps are:
\begin{itemize}
\item $d(x)=x$ and $r(x)=x$ for each $x\in G_0$;
\item $g(hk)=(gh)k$ for all $(g,h)\in G_2$ and $(h,k)\in G_2$;
\item $g d(g)=g$ and $r(g) g=g$ for all $g\in G_1$;
\item $d(g)=g^{-1}g$ and $r(g)= gg^{-1}$ for all $g\in G_1$.
\end{itemize}
The group\-oid $G$ is said to be \emph{open} if the map $d$ is open, which is equivalent to requiring that $r$ is open and also equivalent to requiring that $m$ is open. Then by a \emph{locally compact group\-oid} will be meant an open group\-oid $G$ such that $G_1$ (and hence also $G_0$) is a locally compact space. Finally, a locally compact group\-oid is \emph{second-countable} if $G_1$ (and hence also $G_0$) is second-countable, and it is \emph{sober} if both spaces are sober.
\end{definition}

\begin{example}
Any locally compact group\-oid in the sense of~\cite{Paterson} is a second-countable locally compact group\-oid in the sense of Definition~\ref{def:groupoid}, and so is any Lie group\-oid.
\end{example}

\begin{example}
Let $X$ be a topological space. The \emph{pair groupoid} of $X$, written $\pair(X)$, is the groupoid $G$ defined by
\begin{eqnarray*}
G_1 &=& X\times X\ \text{(with the product topology)}\\
G_0 &=& X
\end{eqnarray*}
such that $d$ and $r$ are the projections $\pi_2$ and $\pi_1$, respectively, $u$ is the diagonal map $X\to X\times X$, and multiplication and inversion are respectively given by
\[
(x,y)(y,z) = (x,z)\quad\text{and}\quad (x,y)^{-1}=(y,x).
\]
\end{example}

\begin{lemma}
Let $G$ be a sober locally compact group\-oid. The topology $\topology(G_1)$, with the Scott topology, is a measurement space --- denoted by $\ms(G)$ ---, and it is classical if and only if $G$ is second-countable (\cf\ Example~\ref{exm:quantales}).
\end{lemma}

\begin{proof}
Since the multiplication map of an open group\-oid is open, the topology $\topology(G_1)$ is a quantale under pointwise multiplication (see~\cite{Re07}). It is also involutive under pointwise involution, and clearly it is strongly Gelfand. Hence, since $G$ is locally compact, $\topology(G_1)$ is a locally compact locale and thus the Scott topology makes it a sober lattice. Moreover, the multiplication preserves joins in each variable, and thus it is continuous (\cf\ comments that precede Lemma~\ref{lem:contscottsober}). So $\topology(G_1)$ is a space of classical measurements if and only if it is countably-based, \ie, if and only if $G$ is a second-countable groupoid.
\end{proof}

A groupoid can be regarded as a space $G_0$ which is equipped with additional local symmetries given by the groupoid structure. For instance, instead of regarding the lab wall $W$ of Example~\ref{exm:labwall} as a local measurement space $\topology(W)$, additional structure may be considered such that $W$ is the object space of a groupoid $G$ with associated classical measurement space $\ms(G)$. (Depending on the type of topology carried by the groupoid, $\topology(W)$ may or may not be an actual subspace of $\ms(G)$ --- \cf\ section~\ref{sec:etale}.)

\subsection{\'Etale groupoids}\label{sec:etale}

Let $Q$ be a stably Gelfand quantale, and $b\in Q$ a \emph{projection}, by which is meant that $b=b^2=b^*$. Then the set
\[
\ipi_b(Q) = \{s\in M\st ss^*\le b,\ s^* s\le b,\ sb\le s,\ bs\le s\},
\]
with the same multiplication as that of $Q$,
is a \emph{pseudogroup}~\cite{SGQ}. Details about pseudogroups will not be needed in this paper, but the definition of $\ipi_b(Q)$ will be important in a few places. We remark that $\ipi_b(Q)$ is a monoid with unit $b$. Note that if $b=1$ we have $\ipi_b(Q)=\ts(Q)$. Especially relevant is that if $Q$ is unital and $b=e$, in which case we write $\ipi(Q)$ instead of $\ipi_e(Q)$, the definition simplifies to
\[
\ipi(Q) = \{s\in M\st ss^*\le e,\ s^* s\le e\}.
\]

\begin{definition}\label{def:IQF}~\cite{Re07,SGQ}
By an \emph{inverse quantal frame} is meant a unital stably Gelfand quantale which is also a locale such that $1=\V\ipi(Q)$. The set of elements $b\le e$ is called the \emph{base locale} and it is denoted by $Q_0$. Then $Q$ is said to be \emph{spatial} if $Q_0$ is a spatial locale, and \emph{locally compact} if $Q_0$ is a locally compact locale (equivalently, $Q$ itself is spatial and locally compact, respectively).
\end{definition}

\begin{example}
Let $X$ be a set. The quantale of binary relations $\mathcal P(X\times X)$ is an inverse quantal frame whose set of partial units consists of the graphs $\operatorname{Graph}(f)$ of bijections $f:U\to V$ for subsets $U,V\subset X$.
\end{example}

A topological groupoid $G$ is said to be \emph{\'etale} if it satisfies the three equivalent conditions of the following proposition:

\begin{theorem}[\cite{Re07}]\label{thm:etale}
Let $G$ be a topological groupoid. The following conditions are equivalent:
\begin{enumerate}
\item The domain map $d:G_1\to G_0$ is a local homeomorphism (equivalently, all the structure maps are local homeomorphisms).
\item $G$ is an open groupoid and $u:G_0\to G_1$ is an open map.
\item $G$ is an open groupoid and the quantale $\topology(G_1)$ is unital with unit $e=u(G_0)$. 
\end{enumerate}
\end{theorem}

\'Etale groupoids are intimately related to inverse quantal frames. The general correspondence requires the notion of \emph{localic} groupoid instead of topological groupoid, but for the purposes of this paper this can be skipped because only spatial (in fact locally compact) inverse quantal frames are relevant. Let us summarize the results we need here:

\begin{lemma}
Let $G$ be a sober locally compact groupoid, and $\ms(G)$ its associated measurement space (which is classical if and only if $G$ is second-countable). The following conditions are equivalent:
\begin{enumerate}
\item $G$ is \'etale.
\item $\ms(G)$ is a (necessarily locally compact) inverse quantal frame.
\end{enumerate}
Moreover, for any locally compact inverse quantal frame $Q$ there is a sober locally compact groupoid $G$ such that $\ms(G)\cong Q$, and $G$ is unique up to a functorial homeomorphism.
\end{lemma}

\begin{proof}
The correspondence between localic \'etale groupoids and inverse quantal frames of~\cite{Re07} specializes to a similar correspondence between sober \'etale groupoids and spatial inverse quantal frames.
\end{proof}

\begin{corollary}
The classical measurement spaces of the form $\ms(G)$ for (second-countable and locally compact) \'etale groupoids $G$ are, up to isomorphism, exactly the same as the countably-based locally compact inverse quantal frames equipped with the Scott topology.
\end{corollary}

If $G$ is a sober locally compact \'etale groupoid, the partial units of $\ms(G)$ are referred to as \emph{local bisections} of $G$, and we also denote their pseudogroup by $\ipi(G)$ instead of $\ipi(\ms(G))$. From here on, given a locally compact inverse quantal frame $Q$, I will write $\groupoid(Q)$ to denote a sober locally compact \'etale groupoid groupoid which is canonically defined by the isomorphism $\ms(\groupoid(Q))\cong Q$.

\begin{example}
Let $X$ be a countable set. The discrete pair groupoid $G:=\pair(X)$ is \'etale, and the classical measurement space $\ms(G)$ is the quantale of binary relations $\mathcal P(X\times X)$ equipped with the Scott topology.\end{example}

\subsection{Example: selective measurements}\label{sec:selective}

\paragraph{Schwinger's picture.}

Schwinger~\cite{Schwinger} proposed a formulation of quantum mechanics in terms of a notion of ``selective measurement,'' based on an idealized notion of physical system for which each physical quantity assumes only a finite set of distinct values. If $A$ is such a physical quantity, $M(a)$ is the measurement that selects those systems whose value of that quantity is $a$ and rejects all others. He also considers a more general type of measurement $M(a',a)$ which disturbs the systems being measured. It selects the systems whose value of $A$ is $a$, after which those systems emerge in a new state for which the value is $a'$. So $M(a)$ is identified with $M(a,a)$, and a natural notion of composition of such measurements is given by a multiplication operation such that
\begin{eqnarray*}
M(a'',a')M(a',a)&=&M(a'',a)\\
\textrm{and}\quad M(a''',a'')M(a',a)&=&0\quad\textrm{if}\quad a''\neq a',
\end{eqnarray*}
where $0$ is the measurement that rejects all systems. (The latter plays a similar role to that of the impossible measurement of measurement spaces.)

Schwinger also defines sums of measurements, for which he gives the following intuitive description:

\begin{quote}
``We define the addition of such symbols to signify less specific selective measurements that produce a subensemble associated with any of the values in the summation, none of these being distinguished by the measurement.''
\end{quote}
The expression ``less specific'' brings to mind the interpretation of the specialization order in measurement spaces, suggesting that a sum of two selective measurements $M(a',a)$ and $M(a''',a'')$ might be represented by a join
\[
M(a',a)\vee M(a''',a'').
\]
However, in Schwinger's formulation the sum is that of a $\CC$-algebra, and it satisfies the axiom
\[
\sum_{a\in X_A} M(a)=1,
\]
where $1$ is said to be the measurement that accepts all systems and $X_A$ is the set of possible values of the physical quantity $A$.

As pointed out by Ciaglia et al~\cite{CIM19-I}, the symbols $M(a',a)$ can be identified with arrows of a discrete group\-oid, namely the pair group\-oid $G_A=\pair(X_A)$,
so Schwinger's collection of measurements can be taken to be the group\-oid algebra $\mathfrak A=\CC G_A\cong M_{\vert X_A\vert}(\CC)$, which has unit $1$.

\begin{remark}\label{rem:acceptingmnot1}
Joins of measurements can still be defined in the measurement space $\Max\mathfrak A$, but note that $\V_{a\in X_A}\langle M(a)\rangle$
is the commutative algebra of diagonal matrices rather than $\langle 1\rangle$.
\end{remark}

If $B$ is another physical quantity there is another pair group\-oid $G_B$, and a corresponding algebra $\mathfrak B=\CC G_B$. If both $A$ and $B$ correspond to complete families of observables we should have $\mathfrak A\cong\mathfrak B$. So $G_B$ can also be embedded into $\mathfrak A$, and we can regard $G_A$ and $G_B$ as two different ``bases'' of $\mathfrak A$.

Finally, it is also possible to have measurements $M(b,a)$, where $a$ is a value of the physical quantity $A$ and $b$ is a value of $B$. An example of such a situation can be a measurement of spin that accepts particles with positive spin along $x$ and changes their state to, say, positive spin along $z$. Ciaglia et al~\cite{CIM19-I} call such measurements Stern--Gerlach measurements.
They further argue that these can be used in order to define a 2-group\-oid whose 1-cells are the selective measurements. This structure summarizes the kinematical aspects of Schwinger's measurements.

\paragraph{Physical quantities with infinitely many values.}

If the physical quantity $A$ has infinitely many values, for instance contained in the real line $X=\RR$, there is more than one choice of associated topological groupoid.
One of the general ideas in this paper is that it is open sets, rather than points, that correspond to observable entities. Hence, a natural generalization of the selective measurements of the form $M(a)$ is to consider instead measurements $M(U)$, where $U$ is an open set of $X$.
Such a measurement is meant to select only those systems that lie in $U$. In particular, $M(X)$ itself is the trivial selective measurement which rejects nothing.

Similarly, the selective measurements $M(a,a')$, which in the discrete case were identified with elements of $X\times X$, are not necessarily ``physical'' and should be replaced by something else. An immediate possibility is to consider the pair groupoid $X\times X$ (with the product topology), which is a Lie groupoid, in which case the measurements $M(a,a')$ ought to be replaced by measurements $M(V)$ with $V$ an open set of $X\times X$. However, doing so means that $X$ itself cannot be identified with an open set of $X\times X$ because the groupoid is not \'etale (the diagonal relation $\Delta_X$ is not open in $X\times X$), so the above mentioned measurements $M(U)$ with $U$ an open set of $X$ are ruled out. This means that any open set $V\subset X\times X$ must contain nonidentity arrows $(x,y)\notin\Delta_X$, in other words that it is impossible to perform a ``pure selection'' of the systems whose value of $A$ lies in $U$ without disturbing the systems that are selected.

Hence, if we want to cater for nonperturbing selective measurements, we must replace the pair groupoid by an \'etale groupoid whose object space is $X$. In order to do this, given two open sets $U,V\subset X$ and a homeomorphism $f:V\to U$, define $M(f)$ to be the selective measurement that selects systems in $V$ and carries them onto $U$ according to the rule specified by $f$.
The collection of all such selective measurements is a pseudogroup acting on $X$, and its associated \emph{groupoid of germs} is \'etale~\cite{MaRe10}.

\section{Observers and quantizations}\label{sec:observers}

Here I shall generalize the definition of observer context of~\cite{fop} by means of a new definition of \emph{observer map}, which generalizes the role played by weakly open surjections of involutive quantales in~\cite{QFB}. The main example is that of the reduced C*-algebra $C_r^*(G)$ of a second-countable locally compact Hausdorff \'etale groupoid $G$, which yields a ``quantization'' of $\ms(G)$ that consists of the measurement space $\Max C_r^*(G)$ equipped with a canonical observer map $p:\Max C_r^*(G)\to\ms(G)$. Generalizing the quantic bundle dictionary of~\cite{QFB}, we see that the type of observer map thus obtained has a close correspondence with the type of groupoid $G$, in particular as regards principal and topologically principal groupoids. But an advantage of the greater generality afforded by observer maps is that these results are no longer restricted to groupoids whose reduced C*-algebras are localizable.

\subsection{Contexts and observer contexts}\label{sec:contexts}

\paragraph{Contexts.}

Let us examine again the idea, described in section~\ref{sec:contobs}, of representing a  space of measurements $M$ on another, $N$, by a continuous map $f:M\to N$,  assuming that both $M$ and $N$ are sober lattices. Adopting the intuitive interpretation of disjunctions put forward in section~\ref{sec:disjunctions}, it is clear that, given $m,m'\in M$, a natural representation of $m\vee m'$ in $N$ is the disjunction $f(m)\vee f(m')$: \emph{the measurement that consists of performing either $m$ or $m'$ is naturally interpreted in $N$ as the measurement that consists of performing either $f(m)$ or $f(m')$}. In addition to this, the interpretation of $0$ as the impossible measurement implies that $f(0)$ should be impossible, too, so $f$ is a homomorphism of sober lattices.

Therefore, for sober lattices the natural notion of context (\cf\ section~\ref{sec:contobs}) should correspond to an embedding of sober lattices $\iota:\opens\to M$. Then, in particular, the image $\iota(\opens)$ is closed under arbitrary joins in $M$, so identifying $\opens$ with $\iota(\opens)$ we are led to the following definition:

\begin{definition}
Let $M$ be a sober lattice. By a \emph{context} in $M$ is meant a subspace $\opens\subset M$ which is closed under arbitrary joins in $M$ (a \emph{sub-sup-lattice} of $M$).
\end{definition}

The fact that this is the same as the image of an embedding of sober lattices follows from the following simple fact about sober spaces:

\begin{lemma}\label{sobersuplattices}
Let $X$ be a sober space whose specialization order is a sup-lattice (for instance a sober lattice). Any subspace $Y\subset X$ which is closed under the formation of arbitrary joins is necessarily sober.
\end{lemma}

\begin{proof}
Let $Y\subset X$ be closed under joins, and let $C\subset Y$ be an irreducible closed set in the relative topology of $Y$. Then the closure $\overline C$ in $X$ is an irreducible closed set of $X$. In order to prove the latter assertion let $C_1$ and $C_2$ be closed sets of $X$ such that $C_1\cup C_2=\overline C$. Then $(C_1\cup C_2)\cap Y=\overline C\cap Y=C$, so $(C_1\cap Y)\cup(C_2\cap Y)=C$. Since $C$ is irreducible in $Y$ we must have $C_1\cap Y=C$ or $C_2\cap Y=C$. Then $\overline C\subset C_1$ or $\overline C\subset C_2$, so $\overline C$ is irreducible and therefore $\overline C=\overline{\{x\}}$ for some $x\in X$. Hence, $\overline C$ is the principal ideal $\downsegment(x)=\{z\in X\st z\le x\}$, and $C$ is closed under joins because it is the intersection $\overline C\cap Y$ of two join-closed sets. Therefore $\V C\in C$, and $C$ is the closure of $\{\V C\}$ in $Y$, so we conclude that every irreducible closed set of $Y$ is the closure of a singleton of $Y$. In addition, $Y$ is $T_0$ because it is a subspace of a $T_0$-space, so $Y$ is sober.
\end{proof}

\begin{definition}
Let $M$ be a sober lattice. A context $\opens\subset M$ is \emph{classical} if the subspace topology makes it a space of classical measurements.
\end{definition}

\begin{example}\label{exm:classicalcontext}
Let $M=\Max M_2(\CC)$ and $\opens_0=\{0,\boldsymbol z^\spindown,\boldsymbol z^\spinup,\boldsymbol z\}$ with the usual order. Then $\opens_0$ is a finite locally compact locale, and it is easy to see that the subspace topology on $\opens_0$ consists of all the upwards-closed sets of $\opens_0$, which for a finite lattice is the Scott topology. So $\opens_0$ is a classical context.
\end{example}

\begin{example}\label{exm:classicalcontext2}
More generally, take again $M=\Max M_2(\CC)$. For each matrix $A=(a_{ij})\in M_2(\boldsymbol 2)$, let $V_A\subset\Max M_2(\CC)$ be the linear space
\[
V_A = \{A'\in M_2(\CC)\st a_{ij}=0\ \Longrightarrow\ a'_{ij}=0\}.
\]
The $2\times 2$ matrices $A$ of 0s and 1s form an involutive quantale (isomorphic to the quantale of binary relations on a set of two elements), and the assignment $A\mapsto V_A$ is an injective homomorphism of involutive quantales whose image is a finite locally compact locale $\opens\subset M$ with the Scott topology. So $\opens$ is a classical context, and it is also an inverse quantal frame of which the context $\opens_0$ of Example~\ref{exm:classicalcontext} is the base locale.
\end{example}

\paragraph{Observer contexts.}

Let $\opens$ be a context of a measurement space $M$, and let $\retraction:M\to \opens$ be a homomorphism of sober lattices that interprets $M$ in $\opens$, such that $\retraction(\omega)=\omega$ for all $\omega\in\opens$, so that the pair $(\opens,\retraction)$ is an ``observer'' (\cf\ section~\ref{sec:contobs}).
For any measurement $m\in M$ think of the retract $\retraction(m)$ as the ``best'' translation of $m$ in terms of the measurements in $\opens$. This does not mean that $\retraction$ should be a homomorphism of measurement spaces, and indeed the main examples to be seen in this section show that in general it is not, the reason being that it does not necessarily preserve composition. The main properties of $\retraction$, along with the rationale behind them, were introduced in~\cite{fop}. Let us recall the corresponding definition here, now adding the new notion of \emph{persistency}:

\begin{definition}
Let $M$ be a measurement space. By an \emph{observer context} of $M$ will be meant a pair $(\opens,\retraction)$ where $\opens\subset M$ is a context and $\retraction:M\to\opens$ is a homomorphism of sober lattices such that for all $\omega\in\opens$ and $m\in M$ we have
\begin{eqnarray*}
\retraction(m^*)&=&\retraction(m)^*\quad\text{(symmetry condition)},\\
\retraction(m\omega)&=&\retraction(m)\omega\quad\text{(preparation condition)}.
\end{eqnarray*}
In addition, the observer context is said to be:
\begin{itemize}
\item \emph{classical} if $\opens$ with the subspace topology is a space of classical measurements;
\item \emph{persistent} if $\retraction(m\omega n) = \retraction(m)\omega \retraction(n)$ for all $m,n\in M$ and $\omega\in\opens$.
\end{itemize}
\end{definition}

The definition of observer context has nontrivial consequences for the underlying context, as the following simple lemma shows.

\begin{lemma}\label{prop:obscontextprops}
Let $M$ be a measurement space, and $(\opens,\retraction)$ an observer context of $M$. Then
$\opens$ is an involutive subquantale of $M$, hence being itself a measurement space. Furthermore, for all $m\in M$ and $\omega\in\opens$ we have
\begin{equation}\label{eq:leftmodule}
\retraction(\omega m)=\omega \retraction(m)\quad\text{(result condition)},
\end{equation}
and if the observer context is persistent then for any interleaved product of generic measurements $m_i$ and measurements $\omega_i$ in $\opens$, of the form
\[
m:=\cdots m_1 \omega_1 m_2 \omega_2 \cdots m_k \omega_k m_{k+1}\cdots,
\]
we obtain
\begin{equation}\label{eq:persistentiteration}
\retraction(m)=\cdots \retraction(m_1) \omega_1 \retraction(m_2) \omega_2 \cdots \retraction(m_k) \omega_k \retraction(m_{k+1})\cdots.
\end{equation}
Moreover, observer contexts can be iterated: if $(\opens',\retraction')$ is an observer context of $\opens$ then $(\opens',\retraction'\circ \retraction)$ is an observer context of $M$.
\end{lemma}

\begin{proof}
Let $\omega,\omega'\in\opens$. Then
\[
\retraction(\omega\omega') = \retraction(\omega) \omega'=\omega\omega',
\]
so $\omega\omega'\in\opens$.
Moreover,
\[
\omega^*=\retraction(\omega)^*=\retraction(\omega^*),
\]
so $\omega^*\in\opens$, and thus $\opens$ is an involutive subquantale (it is closed under joins because it is a context). Then \eqref{eq:leftmodule} is an easy consequence of the involution and of the preparation condition,
\[
\retraction(\omega m)=\bigl(\retraction(\omega m)^*\bigr)^*=\retraction(m^*\omega^*)^*=\bigl(\retraction(m^*)\omega^*\bigr)^*
=\omega \retraction(m),
\]
and \eqref{eq:persistentiteration} is obtained by iterating the persistency axiom and both the preparation and the result conditions. The last statement, concerning iteration of observer contexts, is obvious.
\end{proof}

\begin{remark}
The preparation (resp.\ result) condition means that the retraction of an observer context $(\opens,\retraction)$ is a continuous homomorphism of topological right (resp.\ left) $\opens$-modules.
\end{remark}

\begin{example}
The retraction map $\retraction$ of a classical observer context $(\opens,\retraction)$ of a measurement space $M$ can be used in order to define a continuous map
\[
\basechg:X\to\closedsets(X')
\]
for any other classical context $\opens'\subset M$, where $\topology(X)\cong\opens$, $\topology(X')\cong\opens'$, and $\closedsets(X')$ is the space of closed subspaces of $X'$ with the lower Vietoris topology. This can be regarded as a topological analogue of a ``change of basis'' map that to each ``basis vector'' $x\in X$ assigns a closed set of ``new basis vectors'' $\basechg(x)\subset X'$. See~\cite{fop}.
\end{example}

\begin{example}
Let $A$ be a C*-algebra, $B$ a sub-C*-algebra, and $\Theta:A\to B$ a conditional expectation. Then $(\Max B,\Max \Theta)$ is an observer context of $\Max A$.
\end{example}

\begin{example}\label{exm:classicalobservercontext}
Take $M=\Max M_2(\CC)$, as in Example~\ref{exm:classicalcontext2}. For each linear subspace $V\in M$, let $\retraction(V)$ be the linear subspace
\[
\retraction(V) = \left\{A\in M_2(\CC)\st \forall_{i,j}\bigl(a_{ij}= 0\iff \forall_{A'\in V}(a'_{ij}= 0)\bigr)\right\}.
\]
It can be verified that the classical context $\opens$ of Example~\ref{exm:classicalcontext2}, together with the mapping $\retraction:M\to \opens$ thus defined, is a persistent observer context. A proof of this, in much greater generality, will be given in section~\ref{sec:mainobsexamples}.
\end{example}

\begin{example}\label{exm:classicalobservercontext2}
Let again $M=\Max M_2(\CC)$ and $\opens_0=\{0,\boldsymbol z^\spindown,\boldsymbol z^\spinup,\boldsymbol z\}$ (\cf\ Example~\ref{exm:classicalcontext}). Let $\retraction_0:M\to \opens_0$ be defined in terms of the retraction $\retraction$ of Example~\ref{exm:classicalobservercontext}, as follows:
\[
\retraction_0(V) = \retraction(V)\cap \boldsymbol z.
\]
The pair $(\opens_0,\retraction_0)$ is a classical observer context.
\end{example}

The relation between the observer contexts of the two previous examples is an instance of the following proposition and of the fact that observer contexts can be iterated (Lemma~\ref{prop:obscontextprops}):

\begin{theorem}\label{prop:localobs}
Let $M$ be a measurement space which is a locally compact inverse quantal frame with the Scott topology --- equivalently, $M\cong\ms(G)$ for a sober locally compact \'etale groupoid $G$. Then $(M_0,\retraction)$, where $\retraction=(-)\wedge e$, is an observer context of $M$. Moreover, if $M$ is a classical measurement space (\ie, it is countably-based) then $(M_0,\retraction)$ is a classical observer context.
\end{theorem}

\begin{proof}
The mapping $\retraction=(-)\wedge e:M\to M_0$ is continuous because it preserves joins, and it is of course a retraction. Now let $m\in M$ and $\omega\in M_0$. The properties of inverse quantal frames~\cite{Re07} imply that
\[
m\omega\wedge e=(m\wedge e)\omega,
\]
so $\retraction$ is a homomorphism of $M_0$-modules. Hence, $(M_0,\retraction)$ is an observer context. Now note that $G_0$ is a sober locally compact space, second-countable if $G$ is, so in order to show that $M_0$ is a classical measurement space if $M$ is it remains to be proved that the subspace topology of $M_0$ is the Scott topology. First, since $M_0$ is necessarily sober, its topology is contained in the Scott topology. In order to prove the converse inclusion, let $U$ be a Scott open set of $M_0$, and let us see that it belongs to the subspace topology. Let $U'$ be the upper closure of $U$ in $M$. Then $U\subset M_0\cap U'$. On the other hand, if $m\in M_0\cap U'$ then there exists $n\in U$ such that $n\le m$, so $m\in U$, and we conclude $U=M_0\cap U'$. Also, if $m\in U$ and $m\le n$ then $m=m\wedge e\le n\wedge e\in M_0$, so $n\wedge e\in U$ and thus $U=U'\wedge e$. It remains to be proved that $U'$ is Scott open. Let $D\subset M$ be a directed set such that $\V D\in U'$. Then $D\wedge e=\{d\wedge e\st d\in D\}$ is a directed set of $M_0$, and
\[
\V (D\wedge e)=(\V D)\wedge e\in U,
\]
so there is $d\in D$ such that $d\wedge e\in U$, and thus $d\in U'$. This shows that $U'$ is Scott open, so the Scott open set $U$ belongs to the subspace topology of $M_0$ as intended.
\end{proof}

\subsection{Observers and quantic bundles}\label{sec:obsqubundles}

\paragraph{Observer maps.}

The definition of observer context of a measurement space $M$ consists of a subspace $\opens\subset M$ equipped with a suitable retraction $\retraction:M\to\opens$, but it will be convenient to work with more general embeddings $\iota:\opens\to M$:

\begin{definition}
Let $M$ and $\opens$ be measurement spaces. By an \emph{observer map}
\[
p:M\to\opens
\]
will be meant a pair $p=(\iota,\retraction)$ such that $\iota:\opens\to M$ is a homomorphism of measurement spaces, referred to as the \emph{embedding}, and $\retraction:M\to\opens$, referred to as the \emph{retraction}, is a homomorphism of sober lattices which satisfies
\[
\retraction(m^*)=\retraction(m)^*\quad\text{and}\quad \retraction\bigl(m\,\iota(\omega)\bigr)=\retraction(m)\omega
\]
for all $m\in M$ and $\omega\in\opens$. The observer map is \emph{persistent} if
\[
\retraction\bigl(m\,\iota(\omega)\,n\bigr)=\retraction(m)\omega \retraction(n)
\]
for all $m,n\in M$ and $\omega\in\opens$.
When necessary, in order to distinguish different observer maps, we may denote the components $(\iota,\retraction)$ of $p$ by $(p^\iota,p_\retraction)$. By an \emph{observer} of $M$ will be meant a pair $(\opens,p)$ consisting of a measurement space $\opens$ together with an observer map $p:M\to\opens$.
An observer $(\opens, p)$ is \emph{classical} if $\opens$ is a classical measurement space, and \emph{local} if $\opens$ is a local measurement space.
\end{definition}

\begin{remark}
Every observer context $(\opens,\retraction)$ of $M$ defines an observer $(\opens,p)$ where $p^\iota:\opens\to M$ is the inclusion map and $p_\retraction=\retraction$. Conversely, every observer $(\opens,p)$ defines an observer context $(p^\iota(\opens),p^\iota\circ p_\retraction)$. The terminology for observer contexts and for observers (persistent, classical, local) is consistent.
\end{remark}

\paragraph{\'Etale observers.}

Classical observers $\bigl(\ms(G),p\bigr)$ based on \'etale groupoids cater for weakened notions of persistence that will play a role in the main examples of this paper, to be seen in section~\ref{sec:mainobsexamples}. First we need some more terminology:

\begin{definition}
Let $M$ be a measurement space and $(\opens,p)$ a classical observer such that $\opens$ is an inverse quantal frame. The observer is \emph{\'etale} if for all $\omega\in\opens$ and all $m,n\in M$ such that $p_\retraction(m),p_\retraction(n)\in\ipi(G)$ we have
\[
p_\retraction(m)\, \omega\, p_\retraction(n)\le p_\retraction(m\, p^\iota(\omega)\, n).
\]
Moreover, an \'etale observer is \emph{$\ipi$-stable} if it satisfies the equivalent conditions of the following lemma.
\end{definition}

\begin{lemma}
Let $M$ be a measurement space, $(\opens,p)$ an \'etale observer, and $b=p^\iota(e)$.
The following conditions are equivalent:
\begin{enumerate}
\item\label{qb1} $p_\retraction(s)\in\ipi(\opens)$ for all $s\in\ipi_b(M)$;
\item\label{qb2} $p_\retraction(s)p_\retraction(t)\le p_\retraction(st)$ for all $s,t\in\ipi_b(M)$;
\item\label{qb3} $p_\retraction(s)\omega p_\retraction(t)\le p_\retraction(sp^\iota(\omega)t)$ for all $s,t\in\ipi_b(M)$ and $\omega\in\opens$.
\end{enumerate}
\end{lemma}

\begin{proof}
$\eqref{qb1}\Rightarrow \eqref{qb3}$. Assume that \eqref{qb1} holds and let $s,t\in\ipi_b(M)$ and $\omega\in\opens$. Then, since the observer is \'etale, \eqref{qb3} immediately follows.

$\eqref{qb3}\Rightarrow \eqref{qb2}$. Assume that \eqref{qb3} holds and let $s,t\in\ipi_b(M)$. Then $sbt\le st$, and the following derivation proves \eqref{qb2}:
\[
p_\retraction(st)\ge p_\retraction(sbt)=p_\retraction(sp^\iota(e)t)\ge p_\retraction(s)e p_\retraction(t) = p_\retraction(s)p_\retraction(t).
\]

$\eqref{qb2}\Rightarrow \eqref{qb1}$. Assume that \eqref{qb2} holds and let $s\in\ipi_b(M)$. Then
\[
p_\retraction(s)^*p_\retraction(s)=p_\retraction(s^*)p_\retraction(s)\le p_\retraction(s^*s)\le p_\retraction(b)=p_\retraction(p^\iota(e))= e,
\]
and, similarly, $p_\retraction(s)p_\retraction(s)^*\le e$. Hence, $p_\retraction(s)\in\ipi(\opens)$.
\end{proof}

\paragraph{Increasing, full, and multiplicative observers.}

Let $(\opens,\retraction)$ be an observer context of a measurement space $M$. Suppose that for all $m\in M$ we have $m\le \retraction(m)$; that is, $m$ is always at least as determined as $\retraction(m)$. In other words, the approximation of $m$ by $\retraction(m)$ is allowed to add potential properties but to remove none. Since $\retraction$ increases indeterminacy, such an observer context  will be said to be \emph{increasing}. In terms of general observers this translates to the following definition:

\begin{definition}
Let $(\opens,p)$ be an observer of a measurement space $M$. The observer is
\begin{itemize}
\item \emph{increasing} if $p^\iota\circ p_\retraction\ge\ident_M$ (\ie, $p_\retraction$ is left adjoint to $p^\iota$),
\item \emph{full} if $p^\iota(1_{\opens})=1_M$,
\item \emph{multiplicative} if $p_\retraction(mn)\le p_\retraction(m)p_\retraction(n)$ for all $m,n\in M$.
\end{itemize}
\end{definition}

\begin{lemma}
Every increasing observer is full and multiplicative.
\end{lemma}

\begin{proof}
Let $p=(\iota,\retraction)$. If the observer is increasing $\iota$ preserves arbitrary meets because it is a right adjoint, so $\iota(1_{\opens})=1_M$. In addition, for all $m,n\in M$ we have
\[
\retraction(mn) \le \retraction\bigl(\iota(\retraction(m))\iota(\retraction(n))\bigr)=\retraction\bigl(\iota\bigl(\retraction(m)\retraction(n)\bigr)\bigr)= \retraction(m)\retraction(n). \qedhere
\]
\end{proof}

\begin{corollary}\label{cor:pgiso}
Any multiplicative \'etale observer $(\opens,p)$ of a measurement space $M$ satisfies
\[
p_\retraction(m)\omega p_\retraction(n)=p_\retraction\bigl(m p^\iota(\omega)n\bigr)
\]
for all $\omega\in\opens$ and all $m,n\in M$ such that $p_\retraction(m),p_\retraction(n)\in\ipi(\opens)$.
If furthermore the observer is $\ipi$-stable then it satisfies
\[
p_\retraction(s)\omega p_\retraction(t)=p_\retraction\bigl(s p^\iota(\omega)t\bigr)
\quad\text{and}\quad p_\retraction(s)p_\retraction(t)=p_\retraction(st)
\]
for all $\omega\in\opens$ and all $s,t\in\ipi_{p^\iota(e)}(M)$, and $p^\iota:\ipi(\opens)\to\ipi_{p^\iota(e)}(M)$ is an isomorphism of pseudogroups whose inverse is the restriction of $p_\retraction$ to $\ipi_{p^\iota(e)}(M)$.
\end{corollary}

\paragraph{Quantic bundles.}

Let us relate this to the definition of map of quantales used in~\cite{MP1,OMIQ,SGQ,QFB}, where, in analogy with maps of locales, a \emph{map of involutive quantales} $p:M\to \opens$ is simply defined to be a homomorphism $p^*: \opens\to M$ of involutive quantales, called the \emph{inverse image homomorphism} of $p$. Accordingly, $p$ is a \emph{surjection} if $p^*$ is injective. If $p^*$ has a left adjoint $p_!$, referred to as the \emph{direct image homomorphism}, the map $p$ is said to be \emph{semiopen}, and it is \emph{weakly open} if $p_!$ is a homomorphism of left (equivalently, right) $\opens$-modules. Hence, an increasing observer map $(\iota,\retraction):M\to\opens$ is exactly the same as a weakly open surjection $p:M\to\opens$ such that both $p^*$ and $p_!$ are continuous; the left and right $\opens$-module actions are defined by
\[
\omega\cdot m=\iota(\omega)\,m\quad\text{and}\quad m\cdot\omega=m\,\iota(\omega)
\]
for all $\omega\in\opens$ and $m\in M$. Moreover, if $p^*$ and $p_!$ are continuous and $\opens$ is an inverse quantal frame, this correspondence is completed as follows: a \emph{quantic bundle} $p:M\to\opens$ is the same as an increasing \'etale observer map; an \emph{$\ipi$-stable} quantic bundle is the same as an increasing $\ipi$-stable \'etale observer map; and a \emph{stable quantic bundle} is the same as an increasing persistent \'etale observer map.

In \cite{QFB} a dictionary was presented that compared Fell bundles $\pi:E\to G$ on \'etale groupoids with maps of involutive quantales $p:\Max A\to\opens(G)$, and in particular with the various types of quantic bundles. So the definition of observer map has mathematical interest of its own, since it caters for a generalization of quantic bundles (because the analogue of $p_!$ is available, but free from the requirement of adjointness), and, on the other hand, it brings with it the new requirement of continuity with respect to the topologies of measurement spaces. Below we shall see that part of the dictionary of~\cite{QFB} can be replicated in the more general setting of observer maps, with the advantage of applying to general \'etale groupoids instead of only those that satisfy the restriction of localizability, whose classification is unknown (\cf\ section~\ref{sec:principal}). On the other hand, below we only work with groupoid C*-algebras (of complex valued functions), which is equivalent to restricting to C*-algebras of trivial Fell line bundles. A purely mathematical study of observer maps versus more general Fell bundles is interesting in its own right, but we shall not pursue it in this paper. In particular, we note that any abelian Fell bundle that yields a quantic bundle is necessarily a line bundle, and an interesting question is that of whether a similar fact holds for \'etale observers without the adjointness condition relating $p^\iota$ and $p_\retraction$.

\subsection{Quantum amplitudes and groupoid C*-algebras}\label{quetgrpd}

A natural question is that of how the well known constructions of C*-algebras from groupoids~\cite{Paterson,RenaultLNMath} translate to the language of measurement spaces, and this section provides some answers to this question. Specifically, we shall see that if $G$ is a second-countable locally compact groupoid which is both \'etale and Hausdorff, and if $A$ is the reduced C*-algebra of $G$, then $\Max A$ has, generalizing Example~\ref{exm:classicalobservercontext}, a classical observer $\bigl(\ms(G),p\bigr)$ whose properties mirror properties of $G$. Conversely, $\Max A$ is a ``quantization'' of $\ms(G)$.

\paragraph{Quantum amplitudes.} Let us regard the arrows of a groupoid $G$ as transitions between states of a classical system which is described by a classical measurement space $\ms(G)$, for now assuming that $G$ is finite with the discrete topology. A quantization of this system can be obtained by assigning a probability amplitude $\alpha(g)$ to each transition $g\in G_1$,
such that $\alpha(g)$ coincides, \`a-la Feynman, with the sum of the amplitudes of all the ``trajectories'' that realize $g$:
\[
\alpha(g) = \sum_{g=hk}\alpha(h)\alpha(k).
\]
In other words, the probability amplitudes define a function $\alpha: G_1\to\CC$ that satisfies $\alpha=\alpha*\alpha$ for the convolution product which is given, for all $\alpha,\beta:G_1\to\CC$ and all $g\in G_1$, by
\[
\alpha * \beta (g) = \sum_{g=hk} \alpha(g)\beta(k).
\]
The set of all the functions $\alpha:G_1\to\CC$ equipped with this product is an algebra $C(G)$, isomorphic to the groupoid algebra $\CC G$, and it has an involution given by
\[
\alpha^*(g) = \overline{\alpha(g^{-1})}.
\]
Then $C(G)$ is a natural locus in which to describe algebraically the quantizations of the system.

More generally, constructions of C*-algebras exist for more general locally compact groupoids equipped with Haar systems of measures~\cite{Paterson}. In particular, if $G$ is a second-countable locally compact groupoid that is both Hausdorff and \'etale (for instance an \'etale Lie groupoid), and if $A=C_r^*(G)$ is its reduced C*-algebra, the elements of $A$ can be identified with complex valued functions on $G_1$~\cites{Renault,QFB}, just as in the finite discrete case, and again we may regard $A$ as a locus for describing quantum probability amplitudes. So in the rest of this section I shall assume that $G$ is a groupoid of this type, as usual writing $\ms(G)$ for the associated \'etale measurement space.

\paragraph{Reduced C*-algebras.}

Let $G$ be a second-countable locally compact Hausdorff \'etale group\-oid, and let us recall some facts about its reduced C*-algebra.
I shall write $\supp f$ and $\osupp f$, respectively, for the \emph{support} and the \emph{open support} of a continuous function $f:G_1\to\CC$:
\begin{eqnarray*}
\osupp f &=& \{x\in G_1\st f(x)\neq 0\},\\
\supp f &=& \overline{\osupp f}.
\end{eqnarray*}
The \emph{convolution algebra} of $G$,
\[
C_c(G) = \{f:G\to\CC\st f\textrm{ is continuous and }\supp f\textrm{ is compact}\},
\]
has multiplication and involution defined for all $x\in G_1$ by
\begin{equation}\label{convolution}
f\conv g(x) = \sum_{x=yz} f(y)g(z), \quad f^*(x) = \overline{f(x^{-1})}.
\end{equation}
(Note that the above sum has only finitely many nonzero summands.)
I shall also use the following notation for all open sets $U\subset G_1$:
\[
C_c(U) := \{f\in C_c(G)\st\supp f\subset U\}.
\]

By~\cite{QFB}, the \emph{reduced norm} on $C_c(G)$ can be defined to be the unique C*-norm $\norm\cdot$ such that the following three conditions hold:
\begin{enumerate}
\item $\norm f_\infty\le \norm f$ for all $f\in C_c(G)$;
\item $\norm f_\infty=\norm f$ if $f\in C_c(U)$ for some $U\in\ipi(G)$;
\item The completion $A$ of $C_c(G)$ under this norm continuously extends the inclusion $C_c(G)\to C_0(G)$ to an injective function $A\to C_0(G)$.
\end{enumerate}
The completion $A$ is the \emph{reduced C*-algebra} of $G$, usually denoted $C_r^*(G)$. Due to condition 3 above I will always regard it concretely as an algebra of functions on $G$. Note that if $G$ is compact then $C_r^*(G)$ coincides with $C(G)$ because
\[
C_c(G)\subset C_r^*(G)\subset C_0(G)\subset C(G)=C_c(G).
\]

\begin{lemma}\label{prop:genconv}
Let $G$ be a second-countable locally compact Hausdorff \'etale group\-oid, let $f,g\in C_r^*(G)$, and let $U\in\ipi(G)$.
\begin{enumerate}
\item\label{prop:genconv1} If $g\in \overline{C_c(U)}$ then for all $x\in G_1$
\[
fg(x) = \left\{
\begin{array}{ll}
0 & \textrm{if }d(x)\notin d(U),\\
f(xz^{-1})g(z) & \textrm{if }z\in U\textrm{ and }d(z)=d(x).
\end{array}\right.
\]
\item\label{prop:genconv2} If $f\in\overline{C_c(U)}$ then for all $x\in G_1$
\[
fg(x) = \left\{
\begin{array}{ll}
0 & \textrm{if }r(x)\notin r(U),\\
f(y)g(y^{-1}x) & \textrm{if }y\in U\textrm{ and }r(y)=r(x).
\end{array}\right.
\]
\end{enumerate}
(Note that $y$ and $z$ above are unique.)
\end{lemma}

\begin{proof}
The first statement is a particular case of~\cite{QFB}*{Lemma 3.7}, and the other can be proved in a similar way.
\end{proof}

\subsection{Observers from groupoid C*-algebras}\label{sec:mainobsexamples}

Now we will see that $\Max C_r^*(G)$ has a canonical observer $\bigl(\ms(G),p\bigr)$.

\begin{lemma}\label{prop:p}
Let $G$ be a second-countable locally compact Hausdorff \'etale group\-oid. The mapping $\overline{C_c(-)}:\ms(G)\to\Max C_r^*(G)$ that sends each $U\in\ms(G)$ to the closure $\overline{C_c(U)}$ in $C_r^*(G)$ is a homomorphism of measurement spaces.
\end{lemma}

\begin{proof}
$\overline{C_c(-)}$ is a homomorphism of involutive quantales due to the results of~\cite{QFB}. And, since the topology of $\Max C_r^*(G)$ is contained in the Scott topology, it is continuous because it preserves joins.
\end{proof}

\begin{lemma}\label{prop:osuppcont}
Let $G$ be a second-countable locally compact Hausdorff \'etale group\-oid. The mapping
$\osupp : C_r^*(G)\to\ms(G)$
is continuous and sup-linear.
\end{lemma}

\begin{proof}
Let $f\in C_r^*(G)$, and let $\mathcal U\subset\ms(G)$ be a Scott open set such that
\[
\osupp f\in\mathcal U.
\]
From the hypothesis that $G$ is locally compact Hausdorff it follows that $\osupp f$ is a directed union of the open sets $V$ such that $\overline V$ is compact and $\overline V\subset\osupp f$, so for some of these open sets we must have $V\in\mathcal U$. Let us fix one such $V$ and define
\[
\varepsilon=\min\{\vert f(x)\vert\st x\in\overline V\}.
\]
Then $\varepsilon>0$, and in order to prove that $\osupp$ is continuous it suffices to show that for the open ball $B_\varepsilon(f)\subset C_r^*(G)$ we have
\[
\osupp B_\varepsilon(f)\subset\mathcal U.
\]
In order to see this let $g\in B_{\varepsilon}(f)$ in $C_r^*(G)$. Then
\[
\norm{f-g}_{\infty}\le\norm{f-g}<\varepsilon,
\]
so for all $x\in V$ we have $g(x)\neq 0$ because
\[
\varepsilon-\vert g(x)\vert \le \vert f(x)\vert-\vert g(x)\vert\le\vert f(x)-g(x)\vert<\varepsilon.
\]
So $V\subset\osupp g$, and thus $\osupp g\in\mathcal U$ because $\mathcal U$ is upwards closed. This shows that $\osupp$ is continuous. Moreover, notice that: $\osupp\{0\}=\emptyset$; $\osupp \lambda f=\osupp f$ for all $\lambda\neq 0$; $\osupp \lambda f=\emptyset$ if $\lambda=0$; and evidently $\osupp(f+g)\subset\osupp f\cup\osupp g$. So $\osupp$ is sup-linear.
\end{proof}

The following results yield the main examples of classical observers of this paper, respectively generalizing Example~\ref{exm:classicalobservercontext} and Example~\ref{exm:classicalobservercontext2}. I shall use the following notation, where $V\in \Max C_r^*(G)$ for a group\-oid $G$:
\begin{equation}\label{def:oSupp}
\osupp V := \bigcup_{f\in V}\osupp f.
\end{equation}

\begin{theorem}\label{prop:osuppprops}
Let $G$ be a second-countable locally compact Hausdorff \'etale group\-oid. The pair $p:=\bigl(\overline{C_c(-)},\osupp)$ is an observer map
\[
p:\Max C_r^*(G)\to\ms(G),
\]
and the observer $\bigl(\ms(G),p\bigr)$ is \'etale, full and multiplicative (this is referred to as the \emph{associated observer} of $G$).
\end{theorem}

\begin{proof}
In order to show that $p$ is an observer map we need to prove that
\[
\osupp:\Max C_r^*(G)\to\ms(G)
\]
is a homomorphism of sober lattices and that
the following conditions hold for all $U\in\ms(G)$ and $V\in\Max C_r^*(G)$:
\begin{enumerate}
\item\label{osuppretraction} $\osupp \overline{C_c(U)}=U$.
\item\label{osuppprop2} $\osupp V^\inv=(\osupp V)^\inv$.
\item\label{osuppbimorph} $\osupp(V\overline{C_c(U)})=(\osupp V)U$.
\end{enumerate}

First, by Theorem~\ref{smallvsbigcont}, $\osupp:\Max C_r^*(G)\to\ms(G)$ is a homomorphism of sober lattices because, by Lemma~\ref{prop:osuppcont}, $\osupp:C_r^*(G)\to\ms(G)$ is continuous and sup-linear. Now, for proving \eqref{osuppretraction}, let $U\in\ms(G)$. If $f\in \overline{C_c(U)}$ we must have a sequence $(f_i)$ in $C_c(U)$ such that $\lim_i f_i = f$. Hence, if $x\notin U$ we have $f(x)=\lim_i f_i(x) = 0$. This shows that $\osupp f\subset U$. On the other hand, if $x\in U$ there is $f\in C_c(G)$ such that $f(x)\neq 0$ and $\supp f\subset U$ (because $G$ is locally compact Hausdorff), and thus $\osupp \overline{C_c(U)}=U$.

Now observe that if $f\in C_r^*(G)$ we have $f^*(x)=0$ if and only if $\overline{f(x^{-1})}=0$, so $\osupp f^*=(\osupp f)^*$. From here \eqref{osuppprop2} easily follows, using the fact that $\osupp$ preserves joins.

Next let us prove \eqref{osuppbimorph}. For notational convenience I will prove the left $\ms(G)$-module version of it:
\[
\osupp(\overline{C_c(U)}V)=U\osupp V.
\]
It suffices to prove the equation for one dimensional subspaces $V=\langle f\rangle$ because
\[
\osupp(\overline{C_c(U)}V)=\osupp\bigl(\overline{C_c(U)}\V_{f\in V}\langle f\rangle\bigr)=\bigcup_{f\in V}\osupp(\overline{C_c(U)}\langle f\rangle)\]
and
\[U\osupp V = U\osupp\bigl(\V_{f\in V}\langle f\rangle\bigr)=\bigcup_{f\in V} U\osupp\langle f\rangle.
\]
It also suffices to prove it for $U\in\ipi(G)$ because every open set of $G_1$ is a union of local bisections $(U_i)$, and thus if the equation holds for each $U_i$ it also holds for $U$:
\begin{eqnarray*}
\osupp(\overline{C_c(U)}V)&=&\osupp\bigl(\overline{C_c\bigl(\bigcup_i U_i\bigr)}V\bigr)=\bigcup_i\osupp(\overline{C_c(U_i)}V)\\
&=&\bigcup_i U_i\osupp V= U\osupp V.
\end{eqnarray*}
So for all $U\in\ipi(G)$ and $f\in C_r^*(G)$ let us prove
\begin{equation}\label{onedimbim}
\osupp(\overline{C_c(U)}f)=U\osupp f.
\end{equation}
Let $x\in U\osupp f$. Then we obtain $x=yz$ with $y\in U$ and $z\in\osupp f$ for the unique element
\[
y\in r^{-1}(\{r(x)\})\cap U,
\]
where $z=y^{-1}x$. Letting $g\in C_c(U)$ such that $g(y)\neq 0$ we conclude, using Lemma~\ref{prop:genconv}\eqref{prop:genconv2},
\[
x\in\osupp(gf)\subset \osupp(\overline{C_c(U)}f),
\]
and thus
\[
U\osupp f\subset\osupp(\overline{C_c(U)}f).
\]
For the converse inclusion let $x\in\osupp(\overline{C_c(U)}f)$. Since $\overline{C_c(U)}=\V_{g\in C_c(U)}\langle g\rangle$, we obtain
\[
\osupp(\overline{C_c(U)}f)=\bigcup_{g\in C_c(U)}\osupp(gf),
\]
so for some $g\in C_c(U)$ we have $x\in\osupp(gf)$. Again from Lemma~\ref{prop:genconv}\eqref{prop:genconv2} we conclude that
$x=yz$ with $y\in\osupp g$ and $z\in\osupp f$, so
\[
x\in\osupp g\osupp f\subset U\osupp f.
\]
This proves \eqref{onedimbim}.

Now let us prove that the observer is \'etale. Again it suffices to restrict to one dimensional subspaces of $C_r^*(G)$ and local bisections of $G$; that is, we need to prove the inclusion
\begin{equation}\label{frob3}
\osupp f\ U \osupp g\subset\osupp\bigl(f\overline{C_c(U)}g\bigr)
\end{equation}
for all $U\in\ipi(G)$ and all $f,g\in C_r^*(G)$ such that $\osupp f,\osupp g\in\ipi(G)$.
First note that the condition $U\in\ipi(G)$ implies $\overline{C_c(U)}=C_0(U)$. Then $f\in \overline{C_c\bigl(\osupp f\bigr)}$ because $\osupp f\in\ipi(G)$. Similarly, $g\in \overline{C_c\bigl(\osupp g\bigr)}$. So, for all $s\in \overline{C_c(U)}$, if $x\in \osupp f\osupp s\osupp g$ there is a unique decomposition $x=yzw$ such that $y\in\osupp f$, $z\in\osupp s$, and $w\in\osupp g$, and thus $fsg(s)=f(y)s(z)g(z)\neq 0$, so
\[
\osupp f\,\osupp s\,\osupp g\subset\osupp(fsg).
\]
Hence,
\begin{eqnarray*}
\osupp f\ U \osupp g &=& \bigcup_{s\in \overline{C_c(U)}} \osupp f\,\osupp s\,\osupp g\\
&\subset&\bigcup_{s\in \overline{C_c(U)}} \osupp(asb)\\
&=& \osupp(f \overline{C_c(U)}g),
\end{eqnarray*}
thus proving \eqref{frob3}.

Finally, it is obvious that the observer $\bigl(\ms(G),p\bigr)$ is full, and in order to prove that it is multiplicative all we need is to prove that for all $V,W\in\Max C_r^*(G)$ we have
\begin{equation}\label{eq:fullobs}
\osupp(VW)\subset\osupp V\osupp W.
\end{equation}
Let $f,g\in C_r^*(G)$ and $x\in\osupp(f\conv g)$. From the convolution formula in \eqref{convolution} it follows that for some $(y,z)\in G_2$ such that $x=yz$ we must have both $f(y)\neq 0$ and $g(z)\neq 0$, since otherwise we would obtain $f\conv g(x)=0$. So $x\in\osupp f\osupp g$, showing that $\osupp(f\conv g)\subset\osupp f\osupp g$. Then, using again the fact that $\osupp$ preserves joins of $\Max C_r^*(G)$, \eqref{eq:fullobs} is obtained.
\end{proof}

\begin{corollary}
Let $G$ be a second-countable locally compact Hausdorff \'etale group\-oid. There is an observer map
\[
p_0:\Max C_r^*(G)\to\topology(G_0)
\]
which coincides with the pair $\bigl(\overline{C_c(-)},\osupp(-)\cap G_0\bigr)$.
\end{corollary}

\begin{proof}
This follows from Theorem~\ref{prop:osuppprops} and Theorem~\ref{prop:localobs}, where for simplicity we are identifying $G_0$ with its image $u(G_0)$ in $G_1$.
\end{proof}

We also record the following previously not obvious fact:

\begin{corollary}
Let $A$ be a commutative C*-algebra. The lower Vietoris topology on the locale of closed ideals $I(A)$ coincides with the Scott topology.
\end{corollary}

\begin{proof}
The same techniques used for proving Theorem~\ref{prop:osuppprops} show that for any locally compact Hausdorff space $X$ (not necessarily second-countable) the pair $p=(C_0(-),\osupp)$ is an observer map $p:\Max C_0(X)\to\topology(X)$, so the topology on $I(A)=C_0(\topology(X))$ is the Scott topology.
\end{proof}

\subsection{Principality and localizability}\label{sec:principal}

The conditions satisfied by the associated observer map of a groupoid have a close correspondence with conditions satisfied by the groupoid itself. Let us see examples of this.

\paragraph{Principal and topologically principal groupoids.}

Here we see the main example of a persistent classical observer and its close relation to groupoids that are equivalence relations, and also look at $\ipi$-stable \'etale observers.

Recall that a groupoid $G$ is \emph{principal} if for each $x\in G_0$ the \emph{isotropy group} $I_x=\hom_G(x,x)$ is trivial, and \emph{topologically principal} if the set of all $x\in G_0$ with trivial $I_x$ is dense in $G_0$.
An equivalent definition of principal groupoid is that $G$ is principal if its pairing map $\langle d,r\rangle:G\to G_0\times G_0$ is injective, so $G$ can be identified with an equivalence relation and the pairing map defines
a continuous bijection
\begin{equation}\label{pairingmapdr}
\phi:G\to G_0\times_{G_0/G}G_0.
\end{equation}
However, $\phi$ is not necessarily a homeomorphism (so the action of $G$ on $G_0$ does not necessarily define a principal $G$-bundle over $G_0/G$), and  the orbits of $G$ are not necessarily discrete subspaces of $G_0$ (\cf\ \cite{QFB}*{Example 5.6}).

\begin{theorem}\label{thm:principal}
Let $G$ be a second-countable locally compact Hausdorff \'etale groupoid, and let $\bigl(\ms(G),p\bigr)$ be the associated observer.
\begin{enumerate}
\item\label{thm:principal1} If $G$ is principal with discrete orbits then the observer is persistent.
\item\label{thm:principal2} If the observer is persistent then $G$ is principal.
\item\label{thm:Istable} If $G$ is topologically principal then the observer is $\ipi$-stable.
\end{enumerate}
\end{theorem}

\begin{proof}
Conditions \eqref{thm:principal1}, \eqref{thm:principal2} and \eqref{thm:Istable} correspond respectively to theorems 5.8, 5.10 and 6.10 of~\cite{QFB}, but with  stable quantic bundles generalized to persistent observer maps, and $\ipi$-stable quantic bundles generalized to $\ipi$-stable observer maps. The proofs are essentially the same once we replace the inverse image homomorphism $p^*$ and the direct image homomorphism $p_!$ of a quantic bundle $p$ by their analogues $p^\iota=\overline{C_c(-)}$ and $p_\retraction=\osupp$ for an observer map $p$, since the proofs do not depend on the adjunction between $p^*$ and $p_!$.
\end{proof}

\begin{corollary}\label{cor:principal}
Let $G$ be a topologically principal second-countable locally compact Hausdorff \'etale groupoid with reduced C*-algebra $A$. Writing $B=\overline{C_c(G_0)}$, we obtain an isomorphism of pseudogroups
\[
\overline{C_c(-)}:\ipi(G)\stackrel\cong\to\ipi_B(\Max A)
\]
whose inverse is the restriction of $\osupp$ to $\ipi_B(\Max A)$.
\end{corollary}

\begin{proof}
This is a consequence of Corollary~\ref{cor:pgiso} and Theorem~\ref{prop:osuppprops}.
\end{proof}

\paragraph{Localizable groupoids.}

Let $G$ be a second-countable locally compact Hausdorff \'etale group\-oid. By definition, every $f\in C_r^*(G)$ is a limit $\lim_i f_i$ of a sequence $(f_i)$ in $C_c(G)$, but it is not clear whether it is always possible to find such a sequence if we impose \[\supp f_i\subset\osupp f\] for all $i$. If this is possible we say that $f$ is a \emph{local limit} (because it is approached by continuous compactly supported functions \emph{locally} in $\osupp f$). This motivates the following definition:

\begin{definition}
By a \emph{localizable group\-oid} will be meant a second-countable locally compact Hausdorff \'etale group\-oid $G$ such that for all $f\in C_r^*(G)$ we have
\[
f\in \overline{C_c(\osupp f)}.
\]
\end{definition}

\begin{remark}
The terminology ``localizable'' is taken from~\cite{QFB}, where it was applied to C*-completions of convolution algebras of Fell bundles. From the results there it  follows that any (not a
priori Hausdorff) localizable group\-oid $G$ such that $G_0$ is Hausdorff must have $G_1$ Hausdorff.
\end{remark}

The characterization of the class of localizable group\-oids is an open problem. The following are sufficient but not necessary conditions for localizability:

\begin{theorem}[\cite{QFB}]
Let $G$ be a second-countable locally compact Hausdorff \'etale group\-oid.
\begin{enumerate}
\item If $G_0=G_1$ (\ie, $G$ is just a space) then $G$ is localizable.
\item If $G$ is compact then $G$ is localizable.
\end{enumerate}
\end{theorem}

The significance of localizability comes from its relation to increasing observer contexts, as implied by the following simple theorem:

\begin{theorem}
Let $G$ be a second-countable locally compact Hausdorff \'etale group\-oid with associated observer $\bigl(\ms(G),p\bigr)$. The following conditions are equivalent:
\begin{enumerate}
\item $G$ is localizable.\label{localizable1}
\item For all $V\in\Max C_r^*(G)$ we have $V\subset p^\iota(p_\retraction( V))$.\label{localizable2}
\item $p_\retraction:\Max C_r^*(G)\to\ms(G)$ is left adjoint to $p^\iota$.\label{localizable3}
\end{enumerate}
\end{theorem}

\begin{proof}
\eqref{localizable1} $\Longrightarrow$ \eqref{localizable2}: If $G$ is localizable then \eqref{localizable2} holds because for all $f\in V$ 
\[
f\in \overline{C_c(\osupp f)}\subset \overline{C_c(\osupp V)}.
\]

\eqref{localizable2} $\Longrightarrow$ \eqref{localizable1}: The localizability condition $f\in \overline{C_c(\osupp f)}$ results from taking $V=\langle f\rangle$ in \eqref{localizable2}.

\eqref{localizable2} $\iff$ \eqref{localizable3}: $\osupp$ and $\overline{C_c(-)}$ are monotone and the retraction condition provides the counit of the adjunction, so we have an adjunction if and only if \eqref{localizable2} holds, for this provides the unit.
\end{proof}

\begin{remark}
Localizability means that $p^\iota$ is right adjoint to $p_\retraction$, so it implies that $p^\iota$ preserves arbitrary meets:
\begin{equation}
p^\iota(\bigwedge_i U_i)=\bigcap_i p^\iota(U_i)\quad\textrm{for all families }(U_i)\textrm{ in }\ms(G).\label{qb0}
\end{equation}
It is not known whether the mere existence of a left adjoint of $p^\iota$ implies localizability, but it is easy to conclude that if $p^\iota$ has a left adjoint $p_!$ then $p_\retraction\le p_!$.
\end{remark}

\section{Discussion and further work}

A philosophical consequence of the work presented in this paper and~\cite{fop}, if measurement spaces are taken to be foundational building blocks of physical theories, is that it is possible to regard measurements themselves as being fundamental: they are the events from which reality is weaved. Such a viewpoint is consistent with the relational interpretation of quantum mechanics~\cite{Rovelli96} (even though this paper is not presented in a way that highlights this): if we adopt the stance that a physical object ``is'' a measurement space, the structure of observers in that space encodes all the conceivable ways in which the object can relate to other objects. A mathematical depiction of this can be obtained by looking further into appropriate categories of measurement spaces --- not only $\MS$ but also related categories, for instance whose morphisms are (possibly generalized) observer maps.

Although the aim here was mainly to provide precise mathematical translations of the concepts surrounding measurements and their relation to classical information, further validation of its ideas may be sought by using them as a stepping stone for a reconstruction of quantum mechanics that relates to actual quantum theory in a way which is similar to how point-set topology relates to measure theory: measurement spaces are the topological (and logical) layer on which additional measure-theoretic structure needs to be placed in order to account for statistical aspects. For instance, observer maps can be loosely regarded as ``numbers free'' versions of quantum information channels~\cite{Werner} and, hence, the embedding $p_\iota:\opens\to M$ of a classical observer map $p:M\to\opens$ is the corresponding analogue of a POVM.

An investigation of measurement spaces in relation to full fledged quantum mechanics ought to be addressed in future work, and a related direction should examine the appearance of classical geometric structure and Hamiltonian dynamics in the state spaces of classical observers.
In addition to this, immediate pending mathematical problems include: (i)~to improve the understanding of measurement spaces of C*-algebras, in particular with regard to their automorphism groups, as hinted at in section~\ref{sec:Cstar}; (ii)~the characterization of the class of localizable groupoids, including their relation to amenability; (iii)~to test the results of section~\ref{sec:observers} against generalizations that involve locally trivial Fell line bundles as in~\cite{Renault} or non-Hausdorff groupoids as in~\cite{ExePit19v3}.

\vspace*{5mm}
\noindent {\sc
Centro de An\'alise Matem\'atica, Geometria e Sistemas Din\^amicos
Departamento de Matem\'{a}tica, Instituto Superior T\'{e}cnico\\
Universidade de Lisboa\\
Av.\ Rovisco Pais 1, 1049-001 Lisboa, Portugal}\\
{\it E-mail:} {\sf pmr@math.tecnico.ulisboa.pt}

\end{document}